\documentclass{lmcs} 
\pdfoutput=1

\usepackage[utf8]{inputenc}

\usepackage{lastpage}
\lmcsdoi{17}{2}{18}
\lmcsheading{}{\pageref{LastPage}}{}{}%
{Feb.~12,~2020}{May~26,~2021}{}


\usepackage{amsmath, amscd, amssymb, amsthm, latexsym}

\newcommand{\N}{\mathbb{N}}
\newcommand{\Q}{\mathbb{Q}}

\newcommand{\R}{\mathbb{R}}

\newcommand{\Z}{\mathbb{Z}}

\newcommand{\sing}{\texttt{Sing}}

\newcommand{\nl}{\medskip\noindent}

\newcommand \str [1] {$\langle  {#1} \rangle$}
\newcommand \Str [1] {\text{Str}{#1}}
\newcommand \Strem [1] {\text{\em Str}{#1}}
\newcommand\Ls{$\langle  \R, +,{<} ,\Z\rangle$} 
\newcommand\Ss{$\langle  \R, +,{<},1 \rangle$}      

\newcommand\Ns{$\langle  \N, + \rangle$} 

\begin{document}

\title[Theories of real addition with and without a predicate for integers]{Theories of real addition with and without a predicate for integers}

\author[A.B\`{e}s]{Alexis B\`{e}s\rsuper{a}}	
\address{\lsuper{a}Univ. Paris Est Creteil, LACL, F-94010 Creteil, France}	
\email{bes@u-pec.fr}  

\author[C.Choffrut]{Christian Choffrut\rsuper{b}}	
\address{\lsuper{b}IRIF (UMR 8243), CNRS and Universit\'e Paris 7 Denis Diderot, 
	France}	
\email{cc@irif.fr}





\begin{abstract}
	\noindent We show that it is decidable whether or not a relation on the reals definable in the structure $\langle  \mathbb{R}, +,<, \mathbb{Z} \rangle$
	can be defined in the structure $\langle  \mathbb{R}, +,<, 1 \rangle$. This result is achieved by obtaining a topological characterization of 
	$\langle  \mathbb{R}, +,<, 1 \rangle$-definable relations in the family of $\langle  \mathbb{R}, +,<, \mathbb{Z} \rangle$-definable relations  and then by following 
	Muchnik's approach of showing that the characterization of the relation $X$ can be expressed in the logic of 
	$\langle  \mathbb{R}, +,<,1, X \rangle$.
		
	The above characterization allows us to prove that there is no intermediate structure between $\langle  \mathbb{R}, +,<, \mathbb{Z} \rangle$ and $\langle  \mathbb{R}, +,<, 1 \rangle$. We also show that a $\langle  \mathbb{R}, +,<, \mathbb{Z} \rangle$-definable relation
	is $\langle  \mathbb{R}, +,<, 1 \rangle$-definable if and only if its intersection with every $\langle  \mathbb{R}, +,<, 1 \rangle$-definable line 
	is $\langle  \mathbb{R}, +,<, 1 \rangle$-definable. This gives a noneffective but simple characterization of $\langle  \mathbb{R}, +,<, 1 \rangle$-definable relations.
\end{abstract}

\maketitle

\section{Introduction}
\label{sec:introduction}

Consider the structure  \Ss\     of the additive ordered group of reals along with the 
constant $1$. It is well-known that the subgroup $\Z$ of integers  is not first-order-definable in this structure.
Add the predicate $x\in \Z$ resulting in the structure \Ls. 
Our main result shows that given a \Ls-definable relation it is decidable whether or not it is \Ss-definable.

The structure \Ls\  is a privileged area of application of algorithmic verification of properties of reactive and hybrid systems,  where logical formalisms involving reals and arithmetic naturally appear, see e.g \cite{BDEK07,FQSW20,LASH}.  It 
admits quantifier elimination and is decidable as proved independently by 
Miller \cite{Mil01} and Weisfpfenning \cite{Weis99}.  The latter's proof uses reduction to the theories of \str{\Z,+,<} and \Ss. 

There are many ways to come across the structure \Ls,  which highlights its significance. One approach is through automata. Cobham considers a fixed base $r$ and represents integers 
as finite  $r$-digit strings.  A subset $X$ of integers is $r-$recognizable
if there exists a  finite automaton accepting precisely the representations in base $r$ of its elements. 
Cobham's theorem  says that if  $X$ is $r$- and $s$-recognizable for two  multiplicatively
independent values $r$ and $s$ (i.e., for all $i,j>0$ it holds $r^{i}\not=s^{j}$)
then $X$ is definable in Presburger arithmetic, i.e., in \Ns\  \cite{cobham,Pre29}. Conversely,   each  Presburger-definable 
subset of $\N$ is $r$-recognizable for every $r$. This result was extended 
to integer relations of arbitrary arity by Sem\"enov \cite{semenov}.

Now consider  recognizability of sets of reals. As early as  1962 B\"{u}chi interprets 
subsets of integers as characteric functions of reals in their binary representations and
shows the decidability of a structure which is essentially an extension of \Ls, namely 
\str{\R_{+}, <, P, \N}  where $P$ if the set of positive powers of $2$ 
\cite[Thm 4]{Buc62}. 
Going one step further, Boigelot \& al. \cite{BRW1998} consider   reals as infinite strings of digits and use Muller automata to speak of 
$r$-recognizable subsets and more generally of $r$-recognizable relations of reals.  
In the papers \cite{BB2009,BBB2010,BBL09}  the  equivalence was proved  between 
(1)  \Ls-definability, (2) $r$- and $s$-recognizability where the  two bases have distinct primes in their factorization \cite[Thm  5]{BBL09} and (3)
$r$- and $s$-weakly recognizability for two independently multiplicative bases, \cite[Thm  6]{BBL09} (a 
relation is  $r$-weakly recognizable if it is   recognized by some  deterministic Muller automaton in which all states in the same strongly connected component are either final or nonfinal).  Consequently, as far as reals are concerned, definability in \Ls\ compared to recognizability or weak recognizability by automata on infinite strings can be seen as the analog of Presburger arithmetic for integers  
compared to recognizability by automata on finite strings.

A natural issue is to find effective characterizations of subclasses of $r-$reco\-gnizable relations. In the case of relations over integers, Muchnik proved that for every base $r \geq 2$ and  arity $k \geq 1$, it is decidable whether a $r$-recognizable relation $X \subseteq \N^{k}$ is Presburger-definable \cite{Muchnik03} (see a different approach in  \cite{Ler05} which provides a polynomial time algorithm). For relations over reals, up to our knowledge, the only known result is due to Milchior who proved  that it is decidable (in linear time) whether a weakly $r-$recognizable subset of $\R$ is definable in \Ss\  \cite{Milchior17}. Our result provides an effective characterization of \Ss-definable relations within \Ls-definable relations. Our approach is inspired by Muchnik's one, which  consists of 
giving a combinatorical characterization of \Ns-definable relations that can be expressed in \Ns\ itself.

Our result has two interesting corollaries. The first one is that there is no intermediate structure between \Ls\ and \Ss, i.e., if an \Ls-definable relation $X$ is not   \Ss-definable, then $\Z$ is definable in the structure \str{\R,+,<,1,X}.
Along with the property
that \Ls\ is the ``common'' substructure of all recognizable and weakly recognizable relations 
this indicates that this structure is central. The second corollary is a noneffective but simple characterization:
an \Ls-definable relation is \Ss-definable if and only if every intersection with a rational line is \Ss-definable. By
rational we mean any line which is the intersection of hyperplanes defined by equations with rational coefficients.

The reader will be able to observe   that, while our results are obviously related to automata questions, proofs do not use automata at all.

\subsection*{Other related works.} 
Muchnik's approach, namely expressing in the theory of the structure a property of the structure itself, can be used in other settings.
We refer the interested reader to the discussion in  \cite[Section 4.6]{SSV14} and also to \cite{PW00,Bes13,Milchior17} for examples of such structures. 
A similar method was already used in 1966, see \cite[Thm 2.2.]{GS66} where the authors were able to express in Presburger theory whether or not 
a Presburger subset is the Parikh image of a  context-free language.

Let us mention a recent series of results by Hieronymi which deal with expansions of \Ls, and in particular with the frontier of decidability for such expansions, see, e.g., \cite{Hie19} and its bibliography.
Finally, in connection with our result  that there is no intermediate structure between \Ss\ and \Ls, Conant recently proved \cite{Con18} that there is no intermediate structure between \str{\Z,+} and \str{\Z,+,<}. Concerning these two theories, it is decidable whether or not a \str{\Z,+,<}-definable relation is actually \str{\Z,+}-definable, see  \cite{ChoFri}.

Now we give a short outline of our paper. Section \ref{sec:prelim} gathers all the basic on the two specific structures \Ls\ and \Ss, taking advantage of the existence of quantifier elimination which allows us to work with simpler formulas. Section \ref{sec:strata} introduces topological notions. In particular we say that the neighborhood of a point $x\in \R^{n}$
relative to a relation $X\subseteq \R^{n}$   has \emph{strata} if there exists a direction such that the intersection of
all sufficiently small neighborhoods 
around $x$ with $X$ is the trace of a union of lines parallel to the given direction.
This reflects the fact that the relations we work with are defined by finite unions of regions of the spaces delimited by 
hyperplanes of arbitrary dimension. In Section \ref{sec:local-properties} 
we show that when $X$ is \Ss-definable all points (except finitely many which we call singular)
have at least one direction which is a stratum. Section \ref{sec:neighborhoods} studies relations between neighborhoods. 
In Section \ref{sec:caract-effectif} we give a necessary and sufficient condition for a \Ls-definable relation to be 
\Ss-definable, namely 1) it has finitely many singular points and  
2) all intersections of $X$ with arbitrary hyperplanes parallel to $n-1$ axes and having rational components on the remaining axis
are \Ss-definable. Then we show that these properties are expressible in \str{\R,+,<,1,X}. 
In Section \ref{sec:no-intermediate-structure} 
we  show that there is no intermediate structure between \Ls\ and \Ss.
Section \ref{sec:combinatorial-characterization} is devoted to the proof that a \Ls-definable relation is \Ss-definable if and only if every intersection with a rational line is \Ss-definable.

\section{Preliminaries}
\label{sec:prelim}

Throughout this work we assume the vector space $\R^{n}$ is provided with the    metric
$L_{\infty}$ (i.e., $|x|=\max_{1\leq i\leq n} |x_{i}|$). The open ball centered at $x\in \R^{n}$ and of radius $r>0$ is denoted by 
$B(x,r)$. Given $x,y \in \R^n$ we let $[x,y]$ (resp. $(x,y)$) denote the closed segment (resp. open segment) with extremities $x,y$. 
Also we use  notations such as $[x,y)$ or $(x,y]$ for half-open segments.

Let us specify our logical conventions and notations.  We work within first-order predicate calculus with equality.  We confuse formal symbols and their interpretations, except in subsection  \ref{ss:decidability} 
where the distinction is needed. 
We are mainly concerned with the structures  \Ss\ and  \Ls. In the latter structure, $\Z$ should be understood as a unary predicate which is satisfied  by reals belonging to  $\Z$ only - in other words, we deal  with one-sorted structures. Given a structure $\mathcal{M}$ with domain $D$ and $X \subseteq D^n$, we say that $X$ is {\em definable in $\mathcal{M}$}, or {\em $\mathcal{M}$-definable}, if there exists a formula $\varphi(x_1,\dots,x_n)$ in the signature of $\mathcal{M}$ such that $\varphi(a_1,\dots,a_n)$ holds in $\mathcal{M}$ if and only if $(a_1,\dots,a_n) \in X$. 

\medskip

The \Ss-theory  admits quantifier elimination in the following sense, which can be interpreted geometrically
as saying that a \Ss-definable relation is a finite union of closed and open polyhedra.

\begin{thmC}[{\cite[Thm.~1]{FR75}}]
	\label{th:quantifier-elimination-for-R-plus}
	Every formula in  \Ss\  is equivalent to a finite Boo\-lean combination  of inequalities  between linear combinations  of variables with coefficients in $\Z$ (or, equivalently, in $\Q$).
\end{thmC}

\begin{cor}\label{cor:basicSsFacts} \hfill
	\begin{enumerate}
		\item A subset of $\R$ is \Ss-definable if and only if it is a finite union of intervals whose endpoints are rational
		numbers. In particular $\Z$ is not \Ss-definable. 	
		\item For every $n \geq 1$, if $X \subseteq \R^n$ is nonempty and \Ss-definable then $X$ contains an element of $\Q^n$.
	\end{enumerate}
	\end{cor}

\begin{proof}
	$(1)$ is an immediate consequence of Theorem \ref{th:quantifier-elimination-for-R-plus}. For $(2)$ we proceed by induction over $n$. The case $n=1$ follows from $(1)$. For $n > 1$, consider the set $Y=\{x_n \ | \ \exists x_1,\dots,x_{n-1} \  (x_1,\dots,x_n) \in X\}$. The set $Y$ is nonempty  and \Ss-definable by our hypothesis on $X$, thus by the base case of the induction $Y$ contains a rational $q$. Then it suffices to apply the induction hypothesis to the $(n-1)-$ary relation $\{(x_1,\dots,x_{n-1}) \ | \ (x_1,\dots,x_{n-1},q) \in X\}$.
\end{proof}

In the larger structure \Ls\
it is possible to separate the integer (superscript `$I$') and fractional  
(superscript `$F$') parts of the reals as follows.

\begin{thmC}[{\cite{BFL08},\cite[p.~7]{BBL09}}]
	\label{th:separation-integer-fractional}
	Let $X\subseteq \R^{n}$ be definable in \Ls. Then there exists a unique finite union
	\begin{equation}
	\label{eq:integer-fractional}
	X=\bigcup^{K}_{k=1} (X_{k}^{(I)} + X_{k}^{(F)})
	\end{equation}
	where 
	\begin{itemize}
		\item the relations $X_{k}^{(I)}$ are pairwise disjoint subsets of $\Z^n$ and are \str{\Z,+,<}-definable
		\item the  relations   $X_{k}^{(F)}$ are distinct subsets of $[0,1)^{n}$ and are \Ss-definable
	\end{itemize}
\end{thmC}

There is  again a geometric interpretation of \Ls-definable relations as a 
regular (in a precise technical way) tiling of the space by a finite number of tiles which are themselves 
finite unions of polyhedra. As  a consequence, the restriction of a \Ls-definable relation to a bounded subset 
is \Ss-definable, as stated in the following lemma.

\begin{lem}
	\label{le:restriction-to-bounded-domain}
	For every \Ls-definable relation
	$X \subseteq \R^n $, its restriction to a bounded domain $[a_{1},b_{1}]\times \cdots \times [a_{n} ,b_{n}]$ 
	where the $a_{i}$'s and the $b_{i}$'s are rationals, is  \Ss-definable. 
\end{lem}

\begin{proof}
	By Theorem \ref{th:separation-integer-fractional} the relation $X$ is a finite union of the form  $\displaystyle \bigcup^{K}_{k=1} (X_{k}^{(I)} + X_{k}^{(F)})$
	where each $X_{k}^{(I)}\subseteq \Z^n$ 
	is  \str{\Z, <,+}-definable and each 
	$X_{k}^{(F)}\subseteq [0,1)^n$ is \Ss-definable. 
	
	Let $A_{k}, B_{k}\in \Z$ be such that 
	\[ 
	[a_{1},b_{1}]\times \cdots \times [a_{n} ,b_{n}]
	\subseteq [A_{1},B_{1}]\times \cdots \times [A_{n} , B_{n}].
	\] 
	The $a_{i}$'s and the $b_{i}$'s are rational thus  \Ss-definable, and 
	the relation $[a_{1},b_{1}]\times \cdots \times [a_{n} ,b_{n}]$ is \Ss-definable as well.
	Now the finite subset 
	\[ 
	T_{k}=X_{k}^{(I)} \cap([A_{1},B_{1}]\times \cdots \times [A_{n} , B_{n}]) \subseteq \Z\times \cdots \times \Z
	\]
	is \Ss-definable, therefore so is the sum $T_{k}+ X_{k}^{(F)}$, and
	also the finite union $S=\displaystyle \bigcup^{K}_{k=1} (T_{k} + X_{k}^{(F)})$. Finally 
	the restriction $X\cap ([a_{1},b_{1}]\times \cdots \times [a_{n} ,b_{n}])
	=S \cap  ([a_{1},b_{1}]\times \cdots \times [a_{n} ,b_{n}])$ is also \Ss-definable.
	\end{proof}

By considering the restriction of the \Ls-relation to a ball containing all possible tiles with their closest neighbors, 
we get  that the neighborhoods of \Ls- and \Ss-definable relations are indistinguishable.

\begin{lem}
	\label{le:Ss-Ls-neighborhoods}
	For every \Ls-definable relation $X \subseteq \R^n$    there exists a  \Ss-definable relation $Y \subseteq \R^n$ 
	such that for all $x\in \R^{n}$ there exists $y\in \R^{n}$ 
	and a real $r>0$ such that the translation $u\mapsto u+y-x$ is a one-to-one mapping
	between $B(x,r)\cap X $ and  $B(y,r)\cap Y$.

\end{lem}

\begin{proof} 
	Let $X$ be a \Ls-definable relation 
	\[
	X=\bigcup^{K}_{k=1} (X_{k}^{(I)} + X_{k}^{(F)})
	\]
	as in expression (\ref{eq:integer-fractional}).
	Set $C =[-1,2)^{n}\subseteq \R^{n}$. Observe that the set 
	\[
	\+P= \{- a  + ((a +C)\cap X) \mid a\in \Z^{n}\} 
	\]
	is finite and contains at most $K^{3^n}$ 
	elements, all of which being subsets of $C$. 
	Thus there exists an integer 
	$N$ such that for all $P\in \+P$ there exists $z\in B(0,N-1) \cap \Z^{n}$
	such that 
	\[
	P= - z  + ((z +C)\cap X).
	\]
	We prove the statement  by defining $Y$ as the restriction of $X$ to $B(0,N)$.
	Indeed, consider an arbitrary  $x\in \R^{n}$ and let
	$x=w+t$ where $w_{i}=\lfloor x_{i} \rfloor$ for $i=1, \ldots, n$. 
	Then the translation $u\mapsto -w+u$ defines a one-to-one correspondence between 
	$(w+C)\cap X$  and 
	$-w+ ((w+C)\cap X)$ which is some $P\in \+P$. By definition of $N$ there exists $z\in B(0,N-1) \cap \Z^{n}$
	such that $z+ P=(z+C)\cap X=(z+C) \cap Y$.
	Then the  translation $\tau(u)= -w+z  + u$ is a one-to-one correspondence between 
	$(w+C) \cap X$ and $(z+C) \cap Y$.
	Since $x$ is interior to $w+C$,
	the point $ \tau(x)=y$ is interior to $z+C\subseteq B(0,N)$, thus for sufficiently
	small $r>0$  the ball $B(x,r)$ is included  in $w+C$ and  the ball $B(y,r)$ 
	is included in $z+C$. Consequently $\tau$ defines
	a one-to-one mapping between $B(x,r)\cap X $ and  $B(y,r)\cap Y $. 
    \end{proof}

\section{Strata}
\label{sec:strata}

The aim is to decide, given $n \geq 1$ and a  \Ls-definable relation $X\subseteq \R^{n}$, whether $X$ is \Ss-definable. Though the relations defined in the two structures have very specific properties (see e.g  \cite{BN88,BBD12} for \Ss-definable relations)
we define  properties that make sense in a setting as general as possible. 
The following clearly defines  an equivalence relation.

\begin{defi}
	\label{de:same-neighborhood}
	Given $x,y \in \R^{n}$ we write  $ x \, {\sim_{X}}\, y$ or simply $ x\sim y$  
	when $X$ is understood, if there exists a real $r>0$ such that 
	the translation $w \mapsto w +y-x$ is a one-to-one mapping from $B(x,r)\cap X$ onto $B(y,r)\cap X$.

\end{defi}

\begin{exa} 
	\label{ex:square}
	Consider a closed subset of the plane delimited by a square. There are 10 equivalence classes:
	the set of points interior to the square, the set of points interior to its complement, the four vertices and the 
	four open edges.
\end{exa}

\begin{defi}
	\label{de:strata}
	\hfill
	\begin{enumerate}
		\item Given  $v \in \R^n$ and a point $y\in \R^n$, let  $L_{v}(y)=\{y+\alpha v \ | \ \alpha \in \R\}$. When $v\not=0$ 
		we say that  $L_{v}(y)$ is  the line passing through $y$ in the direction $v$.
		More generally, if $X\subseteq \R^n$ we let $L_{v}(X)$ denote the set $\bigcup_{x\in X} L_{v}(x)$.
		\item 
		
		A  vector $v \in \R^n$  is an  $X$-\emph{stratum} at $x$ 
		(or simply a \emph{stratum}  when $X$ is understood)
		if there exists a real $r>0$ such that 
		\begin{equation}
		\label{eq:saturation}
		B(x, r)  \cap  L_{v}(X\cap  B(x, r))  \subseteq X.
		\end{equation}
		If $v\not=0$ this can be seen as saying that inside the ball $B(x,r)$, the relation $X$ is a union of lines parallel to $v$. 		
		\item The set of $X$-strata at $x$ is denoted by $\text{Str}_{X}(x)$ or simply $\text{Str}(x)$. 
	\end{enumerate}
	
\end{defi}

\begin{prop}
	\label{pr:strata-subspace}
	For all $X\subseteq \R^{n}$ and $x\in \R^{n}$ the set $\Strem(x) $ is  a vector subspace of 
	$\R^{n}$. 
	
\end{prop}

\begin{proof}
	We start with a lemma.
	
	\begin{lem}
		\label{le:v1-and-v2-saturation}
		Let $v_{1}, v_{2}$ be two  strata at $x$ and let $r>0$ be such that 
		\[ 
		B(x, r)  \cap  L_{v_{1}}(X \cap B(x, r) )  \subseteq X \quad \text{ and } \quad B(x, r)  \cap  L_{v_{2}}(X \cap B(x, r) )  \subseteq X
		\]
		Then 
		\[
		B(x, r)  \cap   L_{v_{1}+  v_{2}}(X\cap B(x,r)) \subseteq X.
		\]
	\end{lem}
	
	\begin{proof} 
		Let $v=v_1+v_2$.	The case when one of the vectors $v, v_1, v_2$
		is null is trivial so we assume $v, v_{1}, v_2\not=0$. We must prove that a point $y\in B(x, r) $ belongs to $X$ if and only if all points 
		of $B(x,r)\cap L_{v}(y)$ do. Consider $z= y+  \lambda v \in  B(x,r)$ with $\lambda\ne 0$. Let $\epsilon>0$ be 
		such that for every point $t\in [y,z]$ the ball $B(t,\epsilon)$ is included in $B(x,r)$
		(such a real exists because 
		the segment $[y,z]$ is compact). Let $n$ be an integer 
		such that $|\frac{1}{n} \lambda  v_{1} |<\epsilon$. Then
		the points $y_{i}=y+\frac{i}{n}\lambda (v_{1} + v_{2})$ for  $0\leq i \leq n$  
		belong to $B(x,r)$ because they lay in the segment $[y,z]$, and due to the 
		choice of $\epsilon$ the points 
		$y_{i}+ \frac{1}{n}\lambda  v_{1}$ 
		for $0\leq i < n$  also   belong to $B(x,r)$.
		Since the vectors $v_{1}$ and $v_{2}$ are strata at  $x$,
		for $0\leq i < n$    we have			
\[
y_{i}  \in X \leftrightarrow y_{i} +  \frac{\lambda}{n}  v_{1}   \in X
\leftrightarrow y_{i} +  \frac{\lambda}{n}  v_{1}  +  \frac{\lambda}{n}  v_{2}  = y_{i+1}
\in X.
\]
 Therefore in particular $z \in X \leftrightarrow y \in X$.
			\end{proof}

\medskip Now we turn to the proof of Proposition  \ref{pr:strata-subspace}.
By definition $v$ is a stratum if and only  if  $\lambda v$ is a stratum for some $\lambda\not=0$. 
Thus it suffices to verify that $\Str(x) $ is closed under addition. If 
$v_{1}$ (resp. $v_{2}$) is a stratum then there exist  $r_{1},r_2>0$
such that
\[ 
B(x, r_{1})  \cap L_{v_{1}}(X\cap B(x, r_{1})) \subseteq X \quad \text{ and } \quad	B(x, r_{2})  \cap L_{v_{2}} (X\cap B(x, r_{2})) \subseteq X.	\]
Thus for $r\leq \min\{r_{1}, r_{2} \}$ we have
\[
B(x, r)  \cap L_{v_{1}}(X\cap B(x, r)) \subseteq X \quad \text{ and } \quad B(x, r)  \cap L_{v_{2}}(X\cap B(x, r))  \subseteq X. \
\]
It then suffices to apply Lemma \ref{le:v1-and-v2-saturation}.	
\end{proof}

\begin{defi}
	\label{de:dimension}
		Let $X\subseteq \R^{n}$ and $x \in \R^n$. The \emph{dimension} dim$(x)$ of $x$ is the dimension of the subspace $\Str(x)$.
\end{defi}

\begin{defi}
	\label{de:singular-point}
	Given  $X\subseteq \R^{n}$, a point $x\in  \R^{n}$ is $X$-\emph{singular}, or simply \emph{singular}, if 
	$\Str(x)$ is trivial, i.e., reduced to the null vector, otherwise it is \emph{nonsingular}.
\end{defi}

\begin{exaC}[(Example \ref{ex:square} continued)]
\label{ex:square2}
  Let $x \in \R^2$. If $x$ belongs to the interior of the square or of its complement, then $\Str(x)= \R^2$. If $x$ is one of the four vertices of the square then 
	we have $\Str(x)=\{0\}$, i.e $x$ is singular. Finally, if $x$ belongs to an open edge of the square but is not a 
	vertex, then $\Str(x)$ has dimension 1, and two  points of  opposite edges have the same one-dimensional subspace, while two points of  adjacent edges  have different  one-dimensional subspaces.
	
\end{exaC}

Note that even non-\Ls-definable relations may have no singular points. Indeed consider in the plane the set $X$ defined as the union of vertical lines at
abscissa $\frac{1}{n}$ for all positive integers $n$. In this case any vertical vector is a stratum at any point of the plane.

\medskip Now it can be  shown that all strata at  $x$ can be defined with respect to a common value $r$ in expression (\ref{eq:saturation}).

\begin{prop} 
	\label{pr:uniform-radius}
	Let $X\subseteq \R^{n}$ and $x \in \R^{n}$.   There exists a real $r>0$ such that for every  $v\in \Strem(x)$ 
	we have
	\[
	B(x, r)  \cap L_{v}(X\cap B(x, r)) \subseteq X.
	\]

\end{prop}

\begin{proof} 
	The case when $\Str(x)$ is reduced to $0$ is trivial so we assume that for some $p>0$ 
	the vectors  $v_{1}, \ldots, v_{p}$ form a basis of the vector space $\Str(x)$. There exist
	$r_{1}, \ldots, r_{p}>0$ such that 
	\[
	B(x, r_{i})  \cap L_{v_{i}}(X\cap B(x, r_{i})) \subseteq X.
	\]
	Then for $r= \min \{r_{1}, \ldots, r_{p}\}$ we have
	\[
	B(x, r)  \cap L_{v_{i}}(X\cap B(x, r)) \subseteq X.
	\]
	Consider an arbitrary vector of $\Str(x)$, say $v=\lambda_{1} v_{1} + \cdots + \lambda_{p} v_{p} $.
	It suffices to apply Lemma \ref{le:v1-and-v2-saturation} successively to 
	$\lambda_{1} v_{1} + \lambda_{2} v_{2} $,  $\lambda_{1} v_{1} + \lambda_{2} v_{2} + \lambda_{3} v_{3}$, 
	\ldots, $\lambda_{1} v_{1} + \cdots + \lambda_{p} v_{p} $.
\end{proof}

\begin{defi}
	
	Let $X\subseteq \R^{n}$ and $x \in \R^{n}$. A \emph{safe radius} (for $x$) is a real $r>0$ satisfying the condition of Proposition \ref{pr:uniform-radius}.
	Clearly if $r$ is safe then so are all $0<s\leq r$. Observe that every real is 
	a safe radius  if   $\Str(x)$ is trivial.
	
\end{defi}

\begin{exaC}[(Example \ref{ex:square} continued)]  For an element $x$ of the interior of the square
	or the interior of its complement,  let $r$ be the (minimal) distance from $x$ to the edges of the square. Then $r$ is safe for $x$. 
	If $x$ is a  vertex  
	then $\Str(x)$ is trivial and every $r>0$ is safe for $x$. In all other cases $r$ is the minimal distance 
	of $x$ to a vertex. 
\end{exaC}

\begin{lem}
	\label{le:sim-and-str}
		Let $X\subseteq \R^{n}$ and $x,y \in \R^n$. If $x\sim y$ then  $\text{\em Str}(x) =\text{\em Str}(y) $.

\end{lem}

\begin{proof}
	For some $r>0$, the translation $u\mapsto u+y-x$ is a one-to-one correspondence between 
	$B(x,r)\cap X$ and $B(y,r)\cap X$. Thus every stratum of $X$ at $x$ is a stratum of $X$ at $y$ and 
	vice versa. 
\end{proof}

The converse of Lemma \ref{le:sim-and-str} is false in general. Indeed consider, e.g.,  $X=\{(x,y)\mid y\leq 0\} \cup \{(x,y)\mid y=1\}$ in $\R^{2}$. The points $(0,0)$ and $(0,1)$ have the same 
subspace of strata, namely that generated by $(1,0)$, but $x\not\sim y$.

\bigskip Now we combine the notions of strata and of safe radius.

\begin{lem}
	\label{le:strx-subset-stry}
	Let $X\subseteq  \R^{n}$, $x\in \R^{n}$  and $r$ be a safe radius for $x$. Then for all $y\in B(x,r)$ we have
	$\Strem{(x)}\subseteq \Strem{(y)}$.
\end{lem}

\begin{proof}
	Indeed,  consider $v\in \Str{(x)}$. For all $s>0$ such that $B(y,s)\subseteq B(x,r)$ we have
		\[
	\begin{array}{ll}
		B(y,s)\cap L_{v}(X\cap B(y,s))& = B(y,s)\cap L_{v}(X\cap B(y,s)\cap B(x,r))  \\
		&\subseteq B(x,r)\cap  L_{v}(X\cap  B(x,r))\subseteq   X.
	\end{array}
	\]	
                \par \vspace{-1.4\baselineskip}\qedhere
\end{proof}

\begin{exaC}[(Example \ref{ex:square} continued)]
\label{ex:square3}
  Consider a point $x$ on an (open) edge
	of the square and a safe radius $r$ for $x$. For every point $y$ in $B(x,r)$ which is not on the edge we have 
	$\Str(x)\subsetneq \Str(y)=\R^{2}$. For all other points we have $\Str(x)=\Str(y)$.
\end{exaC}

We relativize the notion of singularity and strata  to an affine subspace $P\subseteq \R^{n}$. 
The next definition should come as no surprise.

\begin{defi}
	\label{de:H-singular}
	Given an affine subspace $P\subseteq \R^{n}$, a subset $X\subseteq \R^{n} $ and a point $x\in P$, we say that  a  nonzero
	vector $v$ parallel to $P$ 
	is an $(X,P)$-\emph{stratum for the point} $x$  if for all sufficiently small $r>0$ it holds
	\[
	 B(x,r) \cap L_{v}(X \cap B(x,r)  \cap   P)\subseteq X.
	\]
	
	A point $x\in P$ is $(X,P)$-\emph{singular} if it has no $(X,P)$-stratum. For simplicity when $P$ is the space $\R^{n}$ we keep the previous terminology and 
	speak of   $X$-strata and 
	$X$-singular points.
\end{defi}
Singularity and nonsingularity do not go through restriction to affine subpaces. 
\begin{exa} In the real plane, let $X=\{(x,y) \ | \  y<0\}$  and $P$ be the line $x=0$. Then the origin is not $X-$singular but it is $(X,P)-$singular. All other elements of $P$ admit $(0,1)$ as an $(X,P)-$stratum thus they are not $(X,P)-$singular. 
	The opposite situation may occur.  In the real plane, let $X=\{(x,y) \ | \  y<0\} \cup P$ where $P=\{(x,y) \mid x=0\}$.
	Then   the origin is  $X-$singular but it is not  $(X,P)-$singular.
\end{exa}

\section{Local properties}
\label{sec:local-properties}

\subsection{Local neighborhoods}

In this section we recall that  if $X \subseteq \R^n$ is \Ss-definable then the equivalence relation $\sim$ (introduced in Definition 
\ref{de:same-neighborhood}) 
has finite index. This extends easily to the case where $X$ is \Ls-definable.
The claim for \Ss-definable relations can be found, e.g., in \cite[Thm 1]{BN88} (see also \cite[Section 3]{BBD12}) but we 
revisit it to  some extent because of the small 
modifications needed to use it in our setting. 

\medskip  We  define what we mean by ``cones''. 

\begin{defi}
	\label{de:cone}
	Let $\xi=(\xi_1,\dots,\xi_n) \in \R^n$. A
	 \emph{cone with apex $\xi$} is an intersection of finitely many  half\-spaces 
	defined by conditions of the form $u(x-\xi) \triangleleft b$ where  $\triangleleft \in \{<,\leq\}$, $b \in \Q$, and $u$ denotes a linear expression with rational coefficients, i.e., $u(x-\xi)=\sum_{1 \leq i \leq n} a_i (x_i - \xi_i)$ where $a_i \in \Q$.	
\end{defi}

In particular the set reduced to the origin, and the empty set,  
are specific cones in our sense (on the real line they can be described respectively by  
$x\leq 0 \wedge -x\leq 0$ and  $x< 0 \wedge -x< 0$). 

Let $X \subseteq \R^n$ be \Ss-definable. By Theorem \ref{th:quantifier-elimination-for-R-plus} we may assume that 
 $X$ is  defined by a formula 
\begin{equation}
\label{eq:phi}
\phi(x)= \bigvee_{i \in I} C_{i}   \text{ where } C_{i} = 
 \bigwedge_{j\in J_{i}} u_{i,j}(x) \triangleleft_{i,j} b_{i,j} 
\end{equation}
where for all $(i,j)\in I\times J_{i}$ we have $\triangleleft_{i,j} \in \{<,\leq\}$, $b_{i,j} \in \Q$ and $u_{i,j}$ is a linear expression with rational coefficients.

 \medskip Now we associate with  $\phi$  a finite collection of cones.

\begin{defi}
\label{de:local-neighborhoods} Consider all formulas obtained from  expression (\ref{eq:phi}) 
 by replacing in all possible ways  each predicate  
$u_{i,j}(x) \triangleleft_{ij} b_{i,j}$ by  one of the three options
$u_{i,j}(x) \triangleleft_{ij} 0$ or $\texttt{false}$ or $\texttt{true}$. Use the routine simplifications 
so that the resulting formulas are reduced to $\texttt{false}$ or $\texttt{true}$ or are disjunctions of conjunctions 
with no occurrence of  $\texttt{false}$ or $\texttt{true}$.  

Let $\Theta$ be the (finite) set of formulas thus obtained, and let us call \emph{local neighborhood} any relation defined 
by some formula in $\Theta$. In particular  each formula in $\Theta$ defines a finite union of cones of which the origin is an apex.

\end{defi}

In the terminology of  \cite[Thm 1]{BN88} the following
says that an \Ss-definable relation has finitely many ``faces''
which are what we call local neighborhoods.
	
\begin{prop}
	\label{pr:neighborhood-1} Consider an \Ss-definable relation $X$.
	There exists a finite collection $\Theta$ of \Ss-formulas defining  finite unions of cones 
	with apex the origin such that 
	for all $\xi\in \R^{n}$   there exist some $\theta$ in $\Theta$ and some real $s>0$ such that for all $t \in \R^n$ 
	we have 
	\begin{equation}
	\label{eq:local-neighborhood}
	(\theta(t) \wedge |t|<s) \leftrightarrow  (\phi(\xi +t) \wedge |t|<s).
	\end{equation}

\end{prop}

\begin{proof}
Let $\Theta$ be defined as in Definition \ref{de:local-neighborhoods}. Consider the expression (\ref{eq:phi}).  For all $(i,j) \in  I\times J_{i}$ let $A_{i,j}$ denote the hyperplane with equation $u_{i,j}(x)= b_{i,j}$. Let $s >0$ be such that $B(\xi,s)$ intersects only the hyperplanes $A_{i,j}$ which contain $\xi$. For all $(i,j) \in  I\times J_{i}$, if $\xi \in A_{i,j}$ then $u_{i,j}(\xi) = b_{i,j}$ thus for every $t \in \R$ we have $u_{i,j}(\xi+t) \triangleleft_{i,j} b_{i,j}$ if and only if $u_{i,j}(t) \triangleleft_{i,j} 0$. Otherwise if $\xi \not\in A_{i,j}$ then $u_{i,j}(\xi+t) \triangleleft_{i,j} b_{i,j}$ is either always true or always false for $0<t<|s|$.
This shows that for $0<t<|s|$ the formula $\phi(\xi +t)$ is equivalent to a Boolean combination of formulas of the form $u_{i,j}(t) \triangleleft_{i,j} 0$, $\texttt{true}$ or $\texttt{false}$.
\end{proof}

\begin{cor}
	\label{cor:finite-tilda-Ss} 
	Let $X \subseteq \R^n$ be \Ss-definable.
	\begin{enumerate}
		
		\item The equivalence relation 
		$\sim$ has finite index.
		
		\item The set of (distinct) spaces $\Strem(x)$ is finite when $x$ runs over $\R^{n}$.  
		
		\item There exists a fixed finite collection $\mathcal{C}$  of  cones (in the sense of Definition \ref{de:cone}) satisfying the following condition. 
		With each $\sim$-class $E$ is associated a subset ${\mathcal{C'}} \subseteq \mathcal{C}$ such that for every $x \in E$ there exists $r>0$ such that for every $t \in \R^n$ we have
		\[
		(x+t \in X) \wedge |t|<r \ \leftrightarrow \ \big(t \in \bigcup_{C \in \mathcal{C'}} C \big) \wedge |t|<r.
		\]

	\end{enumerate}
\end{cor}

\begin{proof}
	Point  1 follows from the fact that $\Theta$ is finite and that two points $x,y$ which are associated with the same formula $\theta$ in Proposition \ref{pr:neighborhood-1} satisfy $x \sim y$ by definition of $\theta$.  Point  2 is a straightforward consequence of  Point  1 and Lemma \ref{le:sim-and-str}.
	For Point 3 observe that all elements of $E$ must be  associated with equivalent formulas $\theta$ in Proposition \ref{pr:neighborhood-1} and that each formula $\theta \in \Theta$ is a disjunction of formulas which define cones.
\end{proof}

Because of Lemma \ref{le:Ss-Ls-neighborhoods} we have 

\begin{cor}\label{cor:Ls-cones}
	The statements of Corollary \ref{cor:finite-tilda-Ss} extend to the case where $X$ is \Ls-definable.
	
\end{cor}

Combining Corollaries \ref{cor:finite-tilda-Ss} and \ref{cor:Ls-cones} allows 
us to specify properties of singular points for  \Ss- and \Ls-definable relations.

\begin{prop}
	\label{pr:finitely-many-singular-points}
	Let $X \subseteq \R^n$. If $X$ is \Ss-definable then it has finitely many  singular points 
	and their components are rational numbers.
	If $X$ is \Ls-definable then it has a countable number of singular points
	and their components are rational numbers.
	
\end{prop}

\begin{proof} 
	By Proposition \ref{pr:neighborhood-1},  if $\xi$
	is not interior to $X$ or its complement, for small enough $r>0$ the subset $ X$ coincides on 
	$B(\xi,r)$ with a finite nonempty union of 
	open or closed cones of which $\xi$ is an apex.  The boundaries of these cones are 
	hyperplanes $H_{1}, \ldots,  H_{p}$ defined by 
	equations of the form $u_{h,k}(x) = b_{h,k} $ as in the proof of Proposition \ref{pr:neighborhood-1}. If their normals scan a subspace of dimension $p< n$ then the space of strata has dimension at least $n-p$: indeed along  all such directions, the expressions $u_{h,k}(x)$ are constant. 
	Therefore a point is singular only if these normals scan the space $\R^{n}$. There are finitely many  hyperplanes $H_i$, and  
	$n$ hyperplanes whose normals are linearly independent intersect in exactly one point, thus the number of singular points is finite
	and their intersections have  rational components. 
	
	The fact that the set of singular points in an \Ls-definable relation is countable is a direct consequence 
	of the following observation. Let $x\in [a_{1},a_{1}+1)\times \cdots \times [a_{n},a_{n}+1)$
	with $a_{1},\ldots, a_{n}\in \Z$. Then $x$ is $X$-singular if and only if it is $Y$-singular in the restriction 
	\[
	Y=X\cap ([a_{1}-1,a_{1}+1)\times \cdots \times [a_{n}-1,a_{n}+1))
	\]
	because for $r>0$ small enough we have $B(x,r)\cap X =  B(x,r)\cap Y$. By Lemma \ref{le:restriction-to-bounded-domain} each set $Y$ is \Ss-definable thus has finitely many singular points, and there is a countable number of such $Y$'s.
	 \end{proof}

\subsection{Strata in \Ls}

In Proposition \ref{pr:strata-subspace} we proved that the set of strata at a given point is a vector subspace. Here we show more precisely that this subspace has a set of generators consisting of vectors with rational coefficients. 

\begin{prop}
\label{pr:rational-basis}
Let  $X\subseteq \R^{n}$ be a \Ls-definable relation and $\xi \in \R^{n}$. There exists a set of linearly independent 
vectors with rational coefficients generating $\Strem(\xi)$.
 
\end{prop}

\begin{proof}

Because of Lemma 	\ref{le:Ss-Ls-neighborhoods} the collection of local neighborhoods in an \Ls-definable relation 
is identical to the  collection of local neighborhoods in some \Ss-definable relation, thus it suffices to treat the case of 
\Ss-definable relations.

By Proposition
	\ref{pr:neighborhood-1} 
	 for all points $\xi\in \R^{n}$ 
	there exists  a  \Ss-formula $\theta$ 
	defining a finite union of cones with apex $0$ 
		 such that 
	for  some real $s>0$ and for  all $t \in \R^n$ 
	the following condition  is satisfied.
	\begin{equation}
	\label{eq:local-neighborhood-bis}
	(\theta(t) \wedge |t|<s) \leftrightarrow  (\phi(\xi +t) \wedge |t|<s)
	\end{equation}
We give the proof for the case $\xi=0$. The argument can easily be generalized to any $\xi \in \R^n$ using  Expression (\ref{eq:local-neighborhood-bis}).
Given an hyperplane $H$ defined by a linear equation $u(x)=0$, 
we set
\[
H^{\varepsilon} = 
\left\{
\begin{array}{lll}
\{x\mid u(x)=0\} &\text{ if }  \varepsilon &\text{is the symbol } =\\
\{x\mid u(x)<0\} &\text{ if }  \varepsilon &\text{is the symbol } <\\
\{x\mid u(x)>0\} &\text{ if }  \varepsilon  & \text{is the symbol } >
\\
\end{array}
\right.
\]

\begin{lem} 
\label{le:gp-far}
Let $y,z \in \R^n$ and let $H_1,\dots,H_q$ be  hyperplanes in $\R^n$ containing the origin. For all nonzero vectors $v \in \R^n$ which are parallel with  no $H_i$, there exists
 $\alpha \in \R$ and $\varepsilon_1,\dots,\varepsilon_q$ such that the  points $y'=y+\alpha v$ and  $z'=z+\alpha v$ belong to
 $\bigcap_{1 \leq i \leq q} H_i^{\varepsilon_i}$. 

\end{lem}

 \begin{proof} 
 Since $v$ is not parallel with $H_{i}$, for sufficiently large $\alpha_{i}\geq 0$ the points 
 $y+\alpha_{i} v$ and  $z+\alpha_{i} v$ belong to the same halfspace defined by $H_{i}$. It suffices to set  
  $\alpha=\max(\alpha_i)$. 
\end{proof} 

We return to the proof of Proposition \ref{pr:rational-basis}. Expression  (\ref{eq:local-neighborhood-bis}) 
 can be viewed as saying that 
 the relation $Y=\{x \in \R^n \mid \theta(x) \}$  
 satisfies the following condition: there exist $q$ rational  hyperplanes  $H_1,\dots,H_{q}$  such that  
 $Y$ is a finite  union of subsets of the form 
 $
 \bigcap_{1 \leq i \leq q} H_i^{\varepsilon_i}
 $ 
 with  $\varepsilon_i \in \{<,>,=\}$. Among all possible expressions defining $Y$ and involving the hyperplanes 
 $H_1, \dots, H_q$,  choose one where the minimum subset of such  hyperplanes  occurs. Rename if necessary the hyperplanes
as  $H_1, \dots, H_{p}$ with $p\leq q$.  We want to show that
 \[
 \Str(0)=\bigcap_{1 \leq i \leq p} H_i.
 \] 
Clearly, every vector $v$ parallel with all $H_{i}$ is a  stratum for all cones so that $\Str(0)\supseteq \bigcap_{1 \leq i \leq p} H_i$
holds. We prove the opposite inclusion. If  $ \Str(0)$ is trivial or $p=0$ 
 we are done. We assume by way of contradiction that for some vector $v\in \Str(0)$, the  subset 
$J\subseteq \{1, \ldots, p\}$ of indices $j$ such that $v$ belongs to $H_{j}$ is proper.   If  $J=\emptyset$,  by Lemma 
 \ref{le:gp-far} for all points $y,z\in \R^{n}$ there exist  $\alpha \in \R$ and $\varepsilon_1,\dots,\varepsilon_{p}$ such that the  points $y'=y+\alpha v$ and  $z'=z+\alpha v$ belong to
 $\bigcap_{1 \leq i \leq p} H_i^{\varepsilon_i}$. This implies $y'\in Y \leftrightarrow z'\in Y$, and since $v$ is a stratum, we get
 \[
 y\in Y \leftrightarrow y'\in Y \leftrightarrow z'\in Y  \leftrightarrow z\in Y
 \]
 thus $Y=\R^{n}$ which is defined by $\theta(x)=\texttt{true}$ and violates the minimality of $p$.

Now we deal with    $J\not=\emptyset$. By possibly renaming the hyperplanes we assume $J=\{r+1, \ldots, p\}$ with $r\geq 1$. We will show that the   hyperplanes $H_{1}, \ldots, H_{r} $  are useless,
i.e., that  $Y$ can be written as a finite union of subsets of the form
 $\bigcap_{r<i\leq p} H_i^{\varepsilon_i}$. Given a subset
  $A= \bigcap_{r<i\leq p} H_i^{\varepsilon_i}$ we show that for all 
   points $y,z \in A$ we have  $y \in Y \leftrightarrow z \in Y$.  We apply again Lemma  \ref{le:gp-far}: there exist 
    $\alpha \in \R$ and $\varepsilon_{1}, \ldots, \varepsilon_{r}$  such that $y'=y+\alpha v$ and $z'=z+\alpha v$ 
    belong to $\bigcap_{1 \leq i \leq r} H_i^{\varepsilon_i}$.  
    Since  $y,z \in  A$ and $v \in \bigcap_{r < i \leq p} H_i$, we get  $y',z' \in A$ thus 
$y',z' \in \bigcap_{1 \leq i \leq p} H_i^{\varepsilon_i}$. By definition of $\theta$ we get 
$y' \in Y \leftrightarrow z' \in Y$, and since  $v$ is a stratum we obtain $y \in Y \leftrightarrow z \in Y$.
This contradicts the minimality of $p$.
\end{proof}

\subsection{Application: expressing the singularity of a point in a \Ls-definable relation}

The singularity of a point $x$ is defined as the property that no intersection of $X$
with a ball centered at $x$ 
is  a union
of  lines parallel with a given nonzero direction. This property is not directly expressible within \Ls\  since the natural way would be to 
use  multiplication  on reals, which is not \Ls-definable. 
In order to be able to express the property, we give an alternative 
characterization of singularity which relies on the assumption  that $X$ is \Ls-definable.

\begin{lem}
	\label{le:equivalence-to-singularity}
	Given an \Ls-definable relation $X \subseteq \R^n$  and $x\in \R^{n}$ 
	the following two conditions are equivalent:
	\begin{enumerate}
		\item $x$ is singular.
		\item  for all $r>0$,  there exists  $s>0$ such that for   
		all nonzero vectors $v$ of norm less than $s$, there exist two points $y,z\in B(x,r)$ 
		such that $y=z+ v$ and $y\in X \leftrightarrow z\not\in X$.
		
	\end{enumerate}

\end{lem}

Observe that when $X$ is not \Ls-definable,  the two assertions are no longer equivalent. E.g., 
$\Q$ has only singular points but condition 2 holds for no point in $\R$.

\begin{proof}
	In order to prove the equivalence of the two conditions, we write them formally
	\[
	\begin{array}{rl}
	(1) & \forall  r>0 \ \forall v\in \R^{n}\setminus \{0\}  \ \exists s>0  \  \exists y,z\ (y,z\in B(x,r) \wedge  y=sv +z  \wedge 
	(y\in X \leftrightarrow z\not\in X))\\
	(2) & \forall r>0  \  \exists s>0\  \forall v \ ( 0<|v|< s \rightarrow 	\exists  y,z\ ( y,z\in B(x,r) \wedge y=v +z
	\wedge (y\in X \leftrightarrow z\not\in X))
	\end{array}
	\]

	The implication  $(2)\rightarrow (1)$ is shown by contraposition, i.e.,  $\neg (1)\rightarrow \neg (2)$
	and is a simple application of the rule $\exists u \forall v \psi(u,v) \rightarrow \forall v  \exists u \psi(u,v) $.
	Indeed, if $v_{0}$ satisfies $\neg (1)$ then 
		 for every $s$ the condition $\neg (2)$ is satisfied with any vector $v$ colinear with $v_{0}$ of modulus 
	less than $s$.

	Now we prove $(1)\rightarrow (2)$. We must prove that if the point $x$ is singular, then for all $r>0$ there exists $s>0$ such that for every vector $|v|\leq s$ 
	there exist two points  $y,z$ satisfying
	\begin{equation}
	\label{eq:ii}
	y,z\in B(x,r) \wedge y=z+v \wedge (y\in X \leftrightarrow z\not\in X).
	\end{equation}

	Let $k$ be the number of  disjuncts in the formula defining $X$ (cf. expression (\ref{eq:phi})), which is also an upper bound 
	on the number of cones composing the local neighborhoods at a given point. In order to simplify the notation
	we also assume that 
	the point $x$ is the origin.
	Also, it is clear that condition $(2)$ is satisfied if and only if it is satisfied for $r$ small enough which means that 
	we may assume that the following holds
	\begin{equation}
	\label{eq:gamma}
	B(0,r)\cap X = B(0,r)\cap \+C 
	\end{equation}
	where $\+C$ is the union of the cones at $0$ as defined in Corollary \ref{cor:Ls-cones}. 
	We claim that expression (\ref{eq:ii}) holds when $s$ is set to  ${ r \over 2k+1}$.
	Since $0$ is singular,  for every direction $u$ there exists a line
	$L_{u}(w)$ with $w\in B(0,r)\cap  \+C$ which contains points in  $X$ and points in the complement of $X$, that is 
	\[
	\emptyset \subsetneq  B(0,r)\cap  \+C\cap L_{u}(w)\subsetneq B(0,r)\cap  L_{u}(w).
	\]
	Because $\+C$ is closed 
	under the mappings $z \mapsto  \alpha z$ for all $\alpha>0$, for all $0<\beta\leq 1$ we have
	\[
	\emptyset \subsetneq  B(0,r)\cap  \+C\cap L_{u}(\beta w)\subsetneq B(0,r)\cap  L_{u}(\beta w).
	\]
	By choosing $\beta $ small enough if necessary, it is always possible to assume 
	that the length of $B(0,r)\cap  L_{u}(\beta w)$ equals some $t \geq r$. 
	The intersection of $L_{u}(\beta w)$ with 
	$X$ inside  $B(0,r)$ defines 
	$p$ segments, some possibly of length $0$,  successively included in and disjoint from the cones in $\+C$, with $2\leq p\leq 2k+1$. 
	Let $x_{0},  x_{1},  \ldots, x_{p}$ be  the endpoints of these segments in the order they appear along the line,
	with $x_{0}$ and $x_{p}$ being the intersections with the frontier of the ball $B(0,r)$.
	Their projections 
	over any of the axes of $\R^{n}$ for which the coordinate of  $u$ is maximal determines a
	nondecreasing  sequence of reals $a_{0} \leq   a_{1} \leq \cdots \leq  a_{p}$ such that $a_p-a_0=t$. 
		If $p=2$  then either $a_{1}- a_{0}\geq  \frac{t}{2}\geq s$  or  $a_{2}- a_{1}\geq  \frac{t}{2}\geq s$
	and then 
	for all $s'<s$ we can choose two points  $y \in (x_{0}, x_{1})$  and $z \in (x_{1}, x_{2})$ such that $|y-z|=s'$.
	Now assume $p>2$. We have $a_p-a_0=t\geq r$ thus there exists $0 \leq i<p$ such that $a_{i+1}-a_{i}\geq \frac{r}{p} \geq \frac{r}{2k+1}= s$. If $i<p-1$ and $x_{i+1}=x_{i+2}$, i.e., 
	$a_{i+1}=a_{i+2}$ then for all $s'<a_{i+1}-a_{i}$, and hence  for every $s'<s$,  there exists a point $z\in (x_{i},x_{i+1})$ such that 
	$|x_{i+1}-z|=s'$, and we can set $y=x_{i+1}$. The case where $i > 0$ and $x_{i-1}=x_{i}$ is similar. In all other cases, for all $s'<s$
	we can find $z\in (x_{i}, x_{i+1})$ and $y\in (x_{i+1}, x_{i+2})$ such that $|y-z|={s'}$.
\end{proof}

\section{Relations between  neighborhoods}
\label{sec:neighborhoods}
We illustrate the purpose of this section with a very simple example. We start with a cube sitting in the horizontal plane 
with only one face visible. The rules of the game is that we are given a finite collection of vectors such that 
for all $6$ faces and all $12$ edges it is possible to choose vectors  that generate the vector subspace  
of the smallest affine subspace in which they live. Let the point at the center of the upper face move towards the observer (assuming that this direction 
belongs to the initial collection). It will  eventually hit the upper edge of the visible face. Now let the point move to the left along the
edge (this  direction necessarily exists because of the assumption on the collection). The point will hit the upper left vertex. 
Consequently, in the trajectory the point visits three different $\sim$-classes: that of the points on the open upper face, that of the 
points on the open edge and that of the upper left vertex. Here we investigate the adjacency of such equivalence classes having decreasing 
dimensions. 
Observe that another finite collection of vectors may have moved the point from the center of the upper face directly to the 
upper left vertex.

\medskip  Since two $\sim$-equivalent points  have  the same subspace of strata, 
we let $\Str(E)$ denote the common subspace of all points of a $\sim$-class $E$. Similarly, 
dim($E$) is the common dimension of all the points in $E$.

\subsection{Adjacency}

 Consider the backwards trajectory on the cube as discussed above:
 the point passes from 
an $\sim$-equivalence class of low dimension into an $\sim$-equivalence class of higher dimension
along a direction that is proper to this latter class. This leads to  the notion of adjacency.
For technical reasons we allow a class to be adjacent to itself.

\begin{defi}
	Let  $E$  be a nonsingular  $\sim$- class   and let  $v$ be one of its nonzero strata.
	Given a $\sim$- class $F$,  a point $y\in F$ is $v-$\emph{adjacent} 
	\emph{to}  $E$ if there exists $\epsilon>0$ such that for all $0< \alpha\leq \epsilon$ we have $y+\alpha v\in E$.
	
	A $\sim$- class $F$  is $v$-\emph{adjacent} 
	\emph{to}  $E$ if there exists a point $y\in F$  which is $v$-adjacent with  $E$.
	
\end{defi}

\begin{lem}
	\label{le:congruence} 
If the $\sim$- class $F$  is $v$-adjacent to  the $\sim$- class $E$, all elements of $F$ are $v$-adjacent to  $E$. 

There exists at most one 
	$\sim$-class $E$ such that  $F$  is $v$-adjacent to  $E$.  
\end{lem}

\begin{proof}
	We must show that if $y$ and $z$ belong to $F$ and if $y$ is $v$-adjacent to $E$ 
	then there exists a real $\alpha >0$ such that for all $0<\beta < \alpha$
	we have $y + \beta v \sim  z+ \beta v$.
	Indeed, by definition of $\sim$ there exists $r$ such that the translation $t \mapsto t+z-y$ maps $B(y,r) \cap X$
	onto $B(z,r) \cap X$. For all $\alpha $ satisfying $|\alpha v|  <r $ consider any $s$ satisfying 
	$|\alpha v|  + s <r $. Then the above translation maps $B(y+\alpha v,s) \cap X$
	onto $B(z+\alpha v,s) \cap X$, i.e., $y+\alpha v \sim z+\alpha v$ .
	
	The second claim easily follows from the very definition of $v-$adjacency.
\end{proof}

Observe that for any nonsingular  $\sim$-class $E$ and one of its nonzero strata $v$ there always exists 
a $\sim$-class $v$-adjacent to $E$, namely $E$ itself, 
but  there might be different classes $v$-adjacent to $E$.

\begin{exa} Let $X$ be the union of the two axes of the 2-dimensional plane and $v=(1,1)$
which we assume is one of the chosen  strata of the $\sim$-class $\{(x,y) \mid x\not=0, y\not=0\}$.
	The different classes are: the complement of $X$, the origin $\{0\}$ which is a singular point,
	the horizontal axis  deprived of the origin, and the vertical 
	axis  deprived of the origin. The two latter $\sim$-classes are both $v-$adjacent to  the class $\R^{2}\setminus X$. 
\end{exa}

\subsection{Intersection of a line and equivalence classes}

In this section we describe the intersection of a $\sim$-class $E$ with a line parallel to some  $v \in \Str(E)$.

\medskip With the example of the cube discussed at the beginning of Section \ref{sec:neighborhoods}, a line passing 
through a point $x$ on the upper face  along any of the directions of $\Str(x)$ of dimension $2$ 
intersects an open edge or a vertex  at point $y$. In the former case dim$(y)=1$  and in the latter 
dim$(y)=0$, and in both cases $\Str(y)\subsetneq \Str(x)$.

\begin{lem}\label{le:adherence}
	Let $X \subseteq \R^n$, $E,F$ be two $\sim$-classes, and $v \in \Strem(E)\setminus \{0\}$. Let $y$ be an element of $F$ which is  
	topologically adherent to $L_y(v) \cap E$. Then 
	$\Strem(F) \subseteq  \Strem(E)$.
	
	If  $E,F$ are  different, then
	$\Strem(F) \subseteq  \Strem(E)\setminus \{v\}$ and therefore $\text{dim}(F) <  \text{dim}(E)$ .

\end{lem}

\begin{proof}
	If $E=F$ then clearly $\Str(F) =  \Str(E)$. Thus it suffices to consider the case $F \not= E$. 	
	By hypothesis $B(y,r)$ intersects $E$ for every $r>0$, which yields $\Str(F) \subseteq \Str(E)$
	by Lemma \ref{le:strx-subset-stry}. It remains to prove that $v \not\in \Str(F)$. We show that for every $r>0$ we can find in $B(y,r)$ two elements $z_1,z_2$ such that $z_1 \in X \leftrightarrow z_2 \not\in X$ and $z_1-z_2$ is parallel to $v$.

	Let $r$ be a safe radius for $y$. By hypothesis there exists $y' \in B(y,r) \cap L_y(v) \cap E$. Let $s>0$ be a safe radius for $y'$
	such that $B(y',s)\subseteq B(y,r)$.
	We have $y \not\sim y'$, thus there exists $u$ with $|u|<s$ such that $y+u \in X \leftrightarrow y'+u \not\in X$. We set $z_1=y+u$ and $z_2=y'+u$. Both $z_1$ and $z_2$ belong to $B(y,r)$ by our hypothesis on $u,s$ and $y$. Moreover $z_1-z_2=y-y'$ and $y' \in L_v(y)$ thus $z_1-z_2$ is parallel to $v$. 
\end{proof}

\begin{lem}\label{lem:open-bis}
	Let $X \subseteq \R^n$, $x \in \R^n$ a nonsingular point and $v \in \Strem{(x)}\setminus \{0\}$. 
	There exist $y,z \in L_{v}(x)$ such that $x \in (y,z)$ and every element $w$ of $(y,z)$ satisfies $w \sim x$. 
	
\end{lem}

\begin{proof}
	Indeed, let $r$ be a safe radius for $x$ and $(y,z)$ be an open segment on $L_{v}(x)\cap B(x,r)$ containing $x$. 
	Let $w\in (y,z)$ and  let $t>0$ be any real such that $B(w,t) \subseteq B(x,r)$. 
	We show that the translation 
	$u\mapsto u+w-x$ defines a one-to-one correspondence  from $B(x,t)\cap X$   to $B(w,t)\cap X$. Indeed, 
	let $z'\in B(x,t)\cap X$.
	Since $B(x,r)\cap X$  is a union of lines parallel with $v$, 
	we have $L_{v}(z') \cap  B(w,t) \subseteq   L_{v}(z') \cap  B(x,r) \subseteq X$ 
	implying $z'+ w-x\in B(w,t)\cap X$. Conversely,  for every element
	$u\in B(w,t)\cap X$
	we have $u-  w+x\in B(x,t)\cap X$. 
\end{proof}

Lemmas \ref{le:adherence} and \ref{lem:open-bis} lead to the following.

\begin{cor}\label{cor:open-ter}
	Let $X \subseteq \R^n$, $x \in \R^n$ a nonsingular point, $E$ its $\sim$-class  and let $v \in \Strem(x)\setminus \{0\}$. The set  $L_{v}(x)\cap E$  is a union of disjoint  open segments  
	(possibly unbounded in one or both directions) of   $L_{v}(x)$, i.e., 
	of the form $(y-\alpha v , y+ \beta v)$ with $0< \alpha,\beta\leq \infty$ and  $y\in E$. 
	
	If $\alpha < \infty$ (resp. $\beta < \infty$) then the point $y-\alpha v$ (resp. $y+ \beta v$)
	belongs to a $\sim$-class $F\not=E$ where $F$ is $v$-adjacent   (resp. $(-v)$-adjacent )  to $E$, and 
	$\text{dim} (F)< \text{dim} (E)$. 
\end{cor}

\begin{proof}
	In order to prove the first claim it suffices to show that for every $y \in L_v(x) \cap E$, the maximal segment of $L_v(x)$ which contains $y$ and is included in $E$ is an open segment. Let $0< \alpha,\beta\leq \infty$ be maximal such that $(y-\alpha v , y+ \beta v) \subseteq E$. There exist such values $\alpha,\beta$ by Lemma \ref{lem:open-bis} (applied to $y$).  Now if $\alpha<\infty$ then by maximality of $\alpha$ and Lemma \ref{lem:open-bis} (applied to $y-\alpha v$) we have $y-\alpha v \not\in E$. Similarly if $\beta < \infty$ then  $y+\beta v \not\in E$. 
	
	The second claim of the corollary follows from Lemma \ref{le:adherence}.
\end{proof}

\section{Characterization and decidability}
\label{sec:caract-effectif}

\subsection{Characterization of \Ss\ in \Ls}

In this section we give a characterization of \Ls-definable relations which are \Ss-definable.
A \emph{rational section} of a relation $X\subseteq \R^{n}$ is a relation of the form
\[
X^{(i)}_{c}= X \cap (\R^{i}\times \{c\} \times \R^{n-i-1}) \quad \text{for some } c\in \Q,\  0\leq i< n
\]

\begin{thm}
	\label{th:CNS}
	Let $n \geq 1$ and let $X \subseteq \R^n$ be \Ls-definable. Then $X$ is \str{\R,+,<,1}-definable if and only if the  following 
	two conditions hold
	
	\begin{enumerate}
		\item There exist finitely many $X-$singular points.
		\item Every rational section of $X$ is \Ss-definable.
		
	\end{enumerate} 
\end{thm} 

Observe  that both conditions $(1)$ and $(2)$ are needed. Indeed, the relation $X=\R \times \Z$ is \Ls-definable. It  has no singular point thus it satisfies condition $(1)$, but does not satisfy $(2)$ since, e.g., the rational section $X_0^{(0)}=\{0\} \times \Z$ is not \Ss-definable.	
Now, consider the relation $X=\{(x,x) \ | \ x \in \Z\}$ which is \Ls-definable. It does not satisfy condition $(1)$ since every element of $X$ is singular, but it satisfies $(2)$ because every rational section of $X$ is either empty or equal to the singleton $\{(x,x)\}$ for some $x \in \Z$, thus is \Ss-definable.

\medskip 

The necessity of point 1 follows from Proposition \ref{pr:finitely-many-singular-points}.
That of point 2  results from the fact that all rational constants are \Ss-definable by Theorem \ref{th:quantifier-elimination-for-R-plus}, and moreover that \Ss-definable relations are closed under 
direct product and intersection.

Now we prove  that conditions 1 and 2 are sufficient.  We start with some informal discussion. 
Since $X$ possesses finitely many $\sim$-classes,  Corollary \ref{cor:open-ter} suggests that we 
prove the \Ss-definability of the $\sim$-classes by induction on their dimension. 
The case of classes of dimension $0$ is easy to handle using condition $1$. For the induction step, 
we use the same corollary which asserts that  the intersection of a nonsingular class $E$ with 
a line passing through a point $x$ in the class and parallel to a direction of the class is a union of open segments.  
If the segment containing $x$ is finite or  
semifinite  then one of its adherent point belongs to a class $F$ of lower dimension 
and we can define $E$ relatively to $F$ via the notion of adjacency. However the segment may be infinite and thus 
may  intersect no
other equivalence class. So we consider the canonical subspaces defined below, and will use the fact that every line has an intersection 
with one of these.

Formally we define 
\begin{equation}
\label{eq:quadrants}
\begin{array}{l}
H_{i}=\{(x_{1}, \ldots, x_{n})\in \R^{n}\mid  x_{i}=0\} \quad i\in \{1, \ldots, n\}\\
\displaystyle Q_{I} = \bigcap_{i\in I} H_{i},   \quad 
\displaystyle Q'_I=(Q_I \setminus  \bigcup_{i\in \{1,\ldots, n\} \setminus I} H_{i} )
\text{ for all } \emptyset \subset I\subseteq \{1,\ldots, n\}.
\end{array}
\end{equation}
In particular  $Q_{\{1,  \ldots, n\}}=\{0\}$, $Q'_{\{1,  \ldots, n\}}=\emptyset$ and by convention $Q_{\emptyset}=\R^{n}$.
 The $Q'_{i}$'s are not vector subspaces
but with some abuse of language we call them  \emph{canonical subspaces} 
and write $\text{dim}(Q'_{I})$ to mean $\text{dim}(Q_{I})=n-|I|$. 
Moreover for every $I \ne \{1, \ldots, n\}$ the set $Q'_{I} $ is open in $Q_{I} $, that is 
\begin{equation}
\label{eq:Q-QprIme}
\forall x\in Q'_{I}\  \exists \epsilon\ \forall v\  (v\in  Q_{I} \wedge  |v| <\epsilon \rightarrow x+v\in Q'_{I})
\end{equation}
(indeed,  it suffices to assume $|v_{j}| < |x_{j}|$ for all $j\not\in I$).
Observe that point 2 of the theorem implies that for every $I$  the intersection $X\cap Q_{I}$ (resp. $X\cap Q'_{I}$) is \Ss-definable.
Furthermore the sets $Q_{I}'$ define a partition  of the space. 
We also have  a
 trivial but important property which is implicit in the proof of Lemma \ref{le:descente}.

\begin{rem}
\label{re:along-L-v}
If $x\in Q_{I}'$ and $v$ is a vector in the subspace $Q_{I}$ then for all points on 
$y\in L_{v}(x)$ we have $y\in Q_{J}'$ for some $J\supseteq I$.
\end{rem}

Using the canonical subspaces, the proof below can be seen as describing a trajectory starting from a point $x$ in a $\sim$-class $E$,
traveling along a stratum of $E$ until it reaches a class of lower dimension $F$ (by Corollary \ref{cor:open-ter})  or some canonical subspace.
In the first case it resumes the journey from the new class $F$ on. In the second case it is trapped in
the canonical subspace: it resumes the journey 
by choosing one direction of the subspace until it reaches a new $\sim$- class or a point belonging 
to a proper canonical subspace. Along the journey, either the dimension 
of the new class
or  the dimension of the canonical subspace decreases. The journey stops when the point reaches a $(X,Q_I)-$singular point,  
or the origin which is the least canonical subspace.

\begin{defi}
	\label{eq:principal-directions} 
	Given an \Ss-definable relation $X \subseteq \R^n$, a set of
	 \emph{principal directions} (for $X$) is any finite subset $V$ of vectors such that 
for all  nonsingular 
	points $x\in \R^{n}$ and all canonical spaces $Q'_{I}$,  there exists a subset $V'\subseteq V$ which generates 
	the space  $\Str(x)\cap Q_{I}$. 
	Observe that by Proposition \ref{pr:rational-basis} there is no loss of generality to assume that 
	$V \subseteq \Q^{n}$.

\end{defi} 

By Corollary \ref{cor:Ls-cones}  there exist finitely many distinct spaces $\Strem(x)$ when $x$ runs over $\R^{n}$, thus there exists a set $V$ of principal directions for $X$.

\medskip For every $I \subseteq \{1,\dots,n\}$ and every $\sim$-class $E$ we define
$
E^{(I)}=E \cap Q'_{I }
$.
Observe that $E=\bigcup_{I \subseteq \{1,\dots,n\} } E^{(I)}$ which is a disjoint union.
We know that $X$ is a union of finitely many $\sim$-classes, see Corollary \ref{cor:Ls-cones}. Thus in order to prove that $X$ is \Ss-definable it suffices to prove that all sets $E^{(I)}$ are \Ss-definable.
Consider the (height) function $h$ which assigns to every set $E^{(I)}$ the pair of integers
\[
h(E^{(I)})=(\dim(\Str(E)\cap Q_I),\dim(Q'_I)).
\]
 Given two $\sim$-classes $E,F$ and $I,J \subseteq \{1,\dots,n\}$ we define the (partial) ordering 
$F^{(J)}<E^{(I)}$ if $h(F^{(J)})=(a',b')$ and  $h(E^{(I)})=(a,b)$ with either ($a'<a$ and $b' \leq b$) or ($a' \leq a$ and $b' < b$).

We prove by induction on $(a,b)$ that each $E^{(I)}$ is definable from smaller classes $F^{(J)}$, i.e., that $E^{(I)}$ is definable in the expansion of \Ss\ obtained by adding a predicate for each smaller class $F^{(J)}$.

Let $h(E^{(I)})=(a,b)$. 
If $a=0$ then the elements of $E^{(I)}$ have no nonzero $X-$stratum in $Q_I$, thus they are $(X,Q_I)-$singular,
see Definition \ref{de:H-singular}. 
By point 2, $X\cap Q_I$  is \Ss-definable thus it has 
finitely many such points and they are \Ss-definable. The same holds for $E^{(I)}$
which is a finite union of such points.

If $b=0$ then $Q_I=\{0\}$, thus $E^{(I)}$ is either empty or equal to the singleton $\{0\}$, and in both cases $E^{(I)}$ is \Ss-definable.

\medskip Now assume that $a,b>0$.  
The following  details a single  step of the journey explained above.

\begin{lem}\label{le:descente}
	Let  $I \subseteq \{1,\dots,n \}$ and $x \in Q'_I$. Then $x \in E^{(I)}$ if and only if  there
	exists a principal direction 
	$v\in V\cap \Str{(x)} \cap Q_I$,  
	 elements  $y,z\in \R^{n}$,   some $\sim$-class $F$ and some $J\supseteq I$
	such that the following holds:	
	\begin{enumerate}	
		\item  $x \in (y,z)$ and $y-x=\alpha v$ for some positive real $\alpha$
		\item $(y,z)$ does not intersect any class $G^{(K)}$ such that  $G^{(K)}< E^{(I)}$
		\item
		\begin{enumerate}
			\item either ($y \in F^{(J)}$ where 	$F^{(J)} < E^{(I)}$ and $F$ is $v-$adjacent to $E$)
		\item or ($z \in F^{(J)}$ where $F^{(J)} < E^{(I)}$ and $F$ is $(-v)-$adjacent to $E$).
		\end{enumerate} 
	\end{enumerate}		
\end{lem}

\begin{proof}
We first prove that the conditions are sufficient. We assume w.l.o.g. that condition  3  holds for the case (a). 
By hypothesis $x \in Q'_I$ and $x \in (y,z)$, thus it suffices to prove that  $(y,z) \subseteq E^{(I)}$.  Set $z=y-\gamma v$ and consider the union $U$ of all open segments 
$(y,y-\beta v)$, $\beta>0$,  included in $E$. Observe that $U$ is nonempty since $y$ is $v$-adjacent to $E$. If $U$ contains $(y,z)$ we are done so we assume $U=(y,t)$ with $t=y-\gamma' v$ and $\gamma'<\gamma$ and we let $G^{(K)}$ be the class of $t$. 
We set $h(G^{(K)})=(a',b')$.
By Lemma  \ref{le:adherence} and Corollary  \ref{cor:open-ter}  we have 
$\Str(G)\subseteq \Str(E)\setminus \{-v\}$, and because of Remark \ref{re:along-L-v} it holds  $K\supseteq I$. 

If $K \supsetneq  I$ then $b'<b$, which leads to the inclusions
\[
\Str{(G)} \cap Q_{K} \subseteq \Str{(E)} \cap Q_{K}  \subsetneq  \Str{(E)} \cap Q_I.
\]	
This implies  $a'\leq a$ and thus  $G^{(K)}<E^{(I)}$.

If $I=K$ then $b=b'$ and we show that $a'<a$. 
The inclusion 
$\Str{(G)} \subseteq \Str{(E)}\setminus \{-v\}$ along with the fact that $v \in Q_I = Q_{K}$ implies
$\Str{(G)} \cap Q_{K}\subsetneq  \Str{(E)}  \cap Q_{K}$ 
which leads to $a'<a$.  This implies again $G^{(K)}<E^{(I)}$ contradicting point 2.

\medskip

Now we prove that the conditions are necessary. 
By  hypothesis we have $a\not=0$ thus $\Str{(E)}\cap Q_I \ne \{0\}$, and by Definition \ref{eq:principal-directions} the set $V$ contains a basis of $\Str{(E)}\cap Q_I$. We can choose $v \ne 0$ as any element of this basis.  

By Corollary   \ref{cor:open-ter} and Property (\ref{eq:Q-QprIme}) there exists an open  segment of the line $L_{v}(x)$ containing $x$ and only points in $E^{(I)}$. Consider the union $U$ of all such open segments.
The hypothesis $b>0$ implies $\dim(Q_I)\geq 1$ 
which means that  $L_v(x)$ intersects some hyperplane $H_j$ with $j \not\in I$. This implies that $L_v(x)$ is not a subset of $Q'_I$, and a fortiori not a subset of $E^{(I)}$, hence $U$ is not equal to $L_v(x)$. 
Assume without loss of generality that $U$ has an extremity of the form $x-\alpha v$ for some $\alpha >0$. We set
$y=x-\alpha v$, and  $z=x+\beta v$ where $\beta$ is any positive real such that $[x,x+\beta v) \subseteq E^{(I)}$.

We prove that $y,z$ satisfy the conditions of the lemma. Conditions $(1)$ and $(2)$ are easy consequences of the very definition of $y$ and $z$. We show that condition $(3a)$ holds. We set $y\in F^{(J)}$ and $h(F^{(J)})=(a',b')$. By Remark \ref{re:along-L-v} we have $I \subseteq J$ thus $b' \leq b$, and by Lemma \ref{le:adherence} we have 
$\Str(F) \subseteq \Str(E)$ thus $\Str(F) \cap Q_J = \Str(E) \cap Q_I$, hence $a'\leq a$.

If $a=a'$, i.e $ y\in E$, then $y\not\in Q_{I}$ by definition of $y$ and Property  (\ref{eq:Q-QprIme}). It follows that $I \subsetneq J$ i.e. $b'<b$, thus $h(F^{(J)})<h(E^{(I)})$.

If $b=b'$, i.e $I=J$, then by definition of $y$ and Corollary  \ref{cor:open-ter} we have $E \ne F$, and using again Lemma \ref{le:adherence} we obtain $\Str(F) \subseteq  \Str(E)\setminus \{v\}$, and this yields $v\in (\Str(E)\cap Q_{I})\setminus  (\Str(F)\cap Q_{J})$ i.e. $a'<a$, which shows that $h(F^{(J)})<h(E^{(I)})$.
\end{proof}

We can conclude the proof of Theorem \ref{th:CNS}. By our induction hypothesis, every set $F^{(J)}$ such that $F^{(J)}<E^{(I)}$ is \Ss-definable, thus it suffices to prove that $E^{(I)}$ is definable in $\langle  \R, +,<, 1,\+F\rangle$ where 
\[
\+F = \{F^{(J)}\mid F \text{ is } v \text{ adjacent to }  E \text{ for some } v\in V \text{, and } F^{(J)}<E^{(I)} \}.
\]
We use the characterization of $E^{(I)}$ given by Lemma \ref{le:descente} and build a formula which expresses the conditions of this lemma. 
Set  $Z=\{F^{(J)}\mid F^{(J)} < E^{(I)}\}$ 
and let $V'$ be the set of vectors $v\in V$ for which there exist some class $F$ and some subset $J$ such that $F^{(J)} < E^{(I)}$.
A defining formula for $E^{(I)}$ is 
\[
\chi(x):  x \in Q'_I \wedge \bigvee_{v\in V'} \chi_{v}(x)
\]
where $ \chi_{v}(x)$ is defined as follows. Denote by 
$A$ (resp. $B$)  the set of classes $F^{(J)}$ such that  $F^{(J)} < E^{(I)}$ and 
$F$ is $v-$adjacent to $E$ (resp $(-v)-$adjacent to $E$). Then 

\[
\begin{array}{ll}
\chi_{v}(x): & \exists y,z,\alpha,\beta \ (\alpha<0 \wedge \beta>0 \wedge  (y=x+\alpha v \wedge z=x+\beta v)  \\
&\wedge \displaystyle   \forall \gamma (\alpha<\gamma<\beta \rightarrow \bigwedge_{F^{(J)} \in Z} \neg F^{(J)}(x+\gamma v)) \wedge \\
&\displaystyle \wedge (\bigvee_{F^{(J)} \in A} F^{(J)}(y) \vee \bigvee_{F^{(J)} \in B} F^{(J)}(z)) \big).
\end{array}
\]
%

\subsection{Decidability}
\label{ss:decidability}

In this section we prove that it is decidable whether a \Ls-definable relation $X \subseteq \R^n$  is \Ss-definable. 
The main idea is to construct in $(\R,+,<,1,X)$ a sentence $\psi_n$ which expresses the conditions of Theorem \ref{th:CNS},  then use the \Ls-definability of $X$ to re-write $\psi_n$ as a \Ls-sentence, and conclude thanks to the decidability of \Ls. This is an adaptation of Muchnik's technique for Presburger Arithmetic \cite[Theorems 2 and 3]{Muchnik03}. 
In order to simplify the task of constructing $\psi_n$ we re-formulate Theorem \ref{th:CNS}. We extend the notion of
section by allowing to fix several components. A \textit{generalized section of $X$} is a relation of the form
\begin{equation}
\label{eq:s-a}
X_{s,a} =\{(x_{1}, \ldots, x_{n})\in X \mid  x_{s_{1}} =a_{s_{1}}. \ldots, x_{s_{r}} =a_{s_{r}}\}
\end{equation}
where $r>0$, $(s)_{1, \ldots, r}= 1\leq s_{1} < \cdots <s_{r}\leq n$ is an increasing sequence, and  $a(=a_{s_{1}}. \ldots, a_{s_{r}})$ is a $r-$tuple of reals.
When $r=0$ we define $X_{s,a} =X$ by convention. If all elements of $a$ are rationals then $X_{s,a}$ is called a \textit{rational generalized section of $X$}.

\begin{prop}\label{pr:CNS1}
	Let $n \geq 1$ and let 
	$X \subseteq \R^n$ be \Ls-definable. Then $X$ is \str{\R,+,<,1}-definable if and only 
	if every rational generalized section of $X$ has finitely many singular points.
\end{prop}

\begin{proof} If  $X$ is \Ss-definable so is every rational  restriction which therefore has finitely many 
	singular points by point 1 of Theorem \ref{th:CNS}.
	
	We show the opposite direction by decreasing induction of the number $r$ of frozen components  of the rational restriction.
	We use the fact that all rational restrictions are \Ls-definable. If $r=n-1$ the rational generalized section is an \Ls-definable 
	subset of $\R$ with finitely many singular points which implies that it consists of finitely many intervals with rational endpoints
	and we are done by Corollary \ref{cor:basicSsFacts}.
	
	Fix $r>1$ and assume that all rational restrictions  $X_{s,a}$ as in \ref{eq:s-a} with $r$ frozen components are \Ss-definable.
	Consider  a rational generalized section $X_{t,b}$  with $r-1$ frozen components, say 
	\[
	\begin{array}{l}
	(t)_{1, \ldots, r-1}= 1\leq t_{1} < \cdots <t_{r-1}\leq n\\
	b= (b_{t_{1}},  \ldots, b_{t_{r-1}})\quad b_{i}\in \Q, i=1, \ldots, r-1.
	\end{array}
	\]
	It has finitely many singular points by hypothesis. 
	A rational section of 
	$X_{t,b}$  is defined by some  increasing sequence $ (s)_{1, \ldots, r}= 1\leq s_{1} < \cdots <s_{r}\leq n$ and an $r$-tuple 
	$a= (a_{s_{1}},  \ldots, a_{s_{r}})$ of 
	rational numbers 
	such that for some $0\leq u\leq r-1$ we have
	\[
	\begin{array}{l}
	s_{k}=t_{k}, \  k<u,  \quad s_{k+2}=t_{k+1},\   u \leq k\\
	a_{k} =  b_{k} , \  k<u,   \quad a_{k+2} =  b_{k+1} ,\   u \leq k
	\end{array}
	\]
	But then all $X_{t,b}$ are \Ss-definable  by induction
	and so is $X_{s,a}$ by Theorem \ref{th:CNS}.    	
\end{proof}

So far we did not distinguish between formal symbols and their interpretations but here we must
do it if we want to avoid any confusion. 
In order to express that a \Ls-definable $n-$ary relation $X$ is actually \Ss-definable we proceed as follows. 
Let  $\{\+X_{n}(x_1,\dots,x_n) \ | \ n \geq 1\}$ be  a collection of relational symbols. 
We construct a  $\{+,<,1,\+X_{n}\}-$sentence $\psi_{n}(\+X_{n})$ such that 
$\psi_{n}(X_{n})$ holds if and only $X_{n}$ is \Ss-definable.

\begin{prop}\label{prop:autodef}
	Let $\{\+X_{n}(x_1,\dots,x_n) \ | \ n \geq 1\}$ denote a set of relational symbols. 
	For every $n \geq 1$ 
	there exists a $\{+,<,1,\+X_{n}\}-$sentence $\psi_n$   such that for every \mbox{$\{+,<,1,\+X_{n}\}-$}structure ${\mathcal{M}}=(\R,+,<,1,X_{n}) $, if $X_{n}$ is \Ls-definable then  we have $\mathcal{M} \models \psi_n$ if and only if  $X_{n}$ is \Ss-definable.
\end{prop}

\begin{proof}  
		For each $I\subseteq \{1, \ldots, n\}$ we let $\R^{I}$ denote the Cartesian product of copies of $\R$ indexed by $I$
	and for all nonzero reals $r$ and all $x\in \R^{I}$ we set 
	$B_{I}(x,r)=\{y\in \R^{I} \mid |x-y| <r\}$.
	Using Lemma \ref{le:equivalence-to-singularity} we can construct
	the following {$\{+,<,1,\+X_{n}\}-$}formula which expresses the fact that a point $x+y$ where $x\in \R^{I}$ and 
	$y\in \R^{[n]\setminus I}$ is singular, when seen as a point of the generalized section of $X_n$ obtained by freezing to $y$ the 
	components of $[n]\setminus I$ (with some abuse of notation we write $x+y$ for $x \in \R^I$ and $y\in \R^{[n]\setminus I}$):
	\begin{equation}
	\label{eq:I-singular}
	\begin{array}{l}
	\sing_{n,I}(x,y, \+X_{n})\equiv  \forall r\in \R \exists s\in ]0,r[ \  \forall q\in \R^{I}  \  |q| < s \rightarrow \\
	\ \ \ \exists z\in \R^{I}     ( (z, z+q \in B_{I}(x,r)) \wedge (y+z\in \+X_{n} \leftrightarrow y+z+q \notin \+X_{n})).\\
	\end{array}
	\end{equation}
		\nl Now we construct a $\{+,<,1,\+X_n\}-$sentence $\psi_n$ 
	which expresses  the condition of Proposition \ref{pr:CNS1}.
	Some difficulty arises from the fact that we have to express that every \textit{rational} generalized section of $X$ has finitely many singular points, but the set $\Q$ is not \Ls-definable. In order to overcome this issue, we construct $\psi_n$ in such a way that it expresses that \textit{every} generalized section of $X$ has finitely many singular points.
		We define $\psi_n$ as 
		\[
		\psi_n \equiv \bigwedge_{I \subseteq [n]}  \forall y\in \R^{[n]\setminus I} \ \varphi_{n,I}(y)
		\]
	where 	
\begin{equation}
	\label{eq:finitesing}
	\begin{array}{lll}
	\varphi_{n,I}(y)& \equiv & \exists M>0 \ \forall x\in \R^{I}  \ (\sing_{n,I}(x,y, \+X_{n}) \rightarrow |x|<M) \\
	& &  \wedge \ \exists  m>0 \ \forall x,x' \in \R^{I} \\
	& & \quad (( x\ne x' \wedge \sing_{n,I}(x,y, \+X_{n}) \wedge \sing_{n,I}(x',y, \+X_{n})) \rightarrow |x-x'|>m).	 \\
	\end{array}
\end{equation}
The formula $\varphi_{n,I}(y)$ expresses that the generalized section of $X_n$ obtained by freezing to $y$ the 
components of $[n]\setminus I$ has finitely many singular points. 

	We prove first that if $X_{n}$ is \Ls-definable  and satisfies the formula $\psi_{n}$, then it is 
\str{\R,+,<,1}-definable. 
Consider a rational generalized section $X_{s,a}$ of $X$ with 
$(s)_{1, \ldots, r}= 1\leq s_{1} < \cdots <s_{r}\leq n$ and $a=(a_1,\dots,a_r)\in \Q^r$. 
Let $I=\{s_1,\dots,s_r\}$. The sentence $\psi_n$ holds, thus in particular the formula 
 $\varphi_{n,I}(y)$ holds when we assign the $r-$tuple $a$ to the $r$ components of $y$. It follows that $X_{s,a}$ has finitely many singular points, and the result follows from Proposition \ref{pr:CNS1}.
 
 Conversely assume that the \Ls-definable relation $X_{n}$ is  also  \Ss-definable. Then we show that the formula $\psi_n(X_n)$ holds.
Indeed, if this were not the case, then for some $I \subseteq [n]$ the predicate $\forall y\in \R^{[n]\setminus I} \ \varphi_{n,I}(y)$ would be false, i.e  $\varphi_{n,I}(y)$ would be false for some assignment of $y$. This implies that the formula  $\gamma_n(y) \equiv \neg \varphi_{n,I}(y)$ (in which the only free variables are the $(n-|I|)$ variables constituting $y$) defines a nonempty subset $Y$ of $\R^{n-|I|}$. Now $\gamma_n$ is a $\{+,<,1,\+X_n\}-$formula, and $X_n$ is \Ss-definable, thus $Y$ is also \Ss-definable.  By Corollary \ref{cor:basicSsFacts}(2), $Y$ contains a $(n-|I|)-$tuple $q$ of rational elements. Therefore the formula $ \neg \varphi_{n,I}(y)$ holds when we assign the value $q$ to $y$, and this implies that there exists a rational generalized section of $X$ which has infinitely many singular points, and by Proposition \ref{pr:CNS1} this leads to a contradiction.
\end{proof}

\begin{thm}\label{th:eff}
	For every $n \geq 1$ and every \Ls-definable relation $X \subseteq \R^n$, it is decidable whether $X$ is  \Ss-definable.
\end{thm}

\begin{proof}
	Assume that $X$ is \Ls-definable by the formula $\phi(x)$. In Proposition \ref{prop:autodef}, if we substitute $\phi(x)$ for every occurrence
	of $x\in \+X_{n}$ in $\psi_{n}$, then we obtain a \Ls-sentence $\Gamma_n$ which holds in \Ls\ if and only if $X$ is \Ss-definable. The result follows from the decidability of \Ls\  \cite{Weis99}.	
\end{proof}

Let us give a fair estimate of the complexity of the decision problem of Theorem \ref{th:eff}. One can derive from  \cite[Section 5]{Weis99} that the known triply-exponential upper bound for the deterministic time complexity of deciding  Presburger Arithmetic sentences \cite{Oppen78} still holds for \Ls. We proved that given $n \geq 1$ and a relation  $X \subseteq \R^n$ which is \Ls-definable by $\phi$, the question of whether $X$ is \Ss-definable amounts to decide whether the sentence $\Gamma_n$ holds in \Ls.

 It is easy to check that the length of $\Gamma_n$ is of the order of $2^n$ times 
the length of $\phi$. Consequently, 
for fixed $n$, the length of $\Gamma_n$ is linear with respect to the one of $\phi$, thus we also get a  triply-exponential upper bound for the deterministic time complexity of our decision problem.

\section{Non-existence of an intermediate structure between \texorpdfstring{\Ss}{(R,+,<,1)}  and \texorpdfstring{\Ls}{(R,+,<,Z)}}
\label{sec:no-intermediate-structure}

Our aim is to prove the following result.

\begin{thm}\label{th:dim1}
	If $X \subseteq \R^n$ is \Ls-definable 
but not \Ss-definable then the set $\Z$ is definable in  \str{\R,+,<,1,X}.
\end{thm}

In other words, for every  $X \subseteq \R^{n}$ which is \Ls-definable, then \str{\R,+,<,1, X} is inter-definable with either \Ls\  or \Ss.

\subsection{Periodicity in $\R$}

\begin{defi}
Consider  $X \subseteq \R$ and $p \in \R\setminus \{0\}$.

Then  $X$  is {\em periodic}
of \emph{period} $p$  (or {\em $p-$periodic})  if for every real $x$ we have
$x\in  X \leftrightarrow x+p\in X$. 

It is {\em ultimately right  $p-$periodic} if  there exists $m \in \R$ such that for every real $x$ with $x \geq m$, we have $x\in X \leftrightarrow x+p\in X$. We say  that $p$ is a \emph{right ultimate period}.

It is {\em ultimately left  $p-$periodic} if  there exists $m \in \R$ such that for every real $x$ with $x\leq m$, we have $x\in X \leftrightarrow x+p\in X$. We say  that $p$ is a \emph{left ultimate period}.
\end{defi}

   Observe that the empty set is $p-$periodic, ultimately right  $p-$periodic,
ultimately left  $p-$periodic for every $p\not=0$.
We apply these  notions and results concerning  \Ss- and \Ls-definable subsets of $\R$.

\begin{prop}
\label{pr:special-ultimately-periodic-sets}
\hfill
 \begin{enumerate}
 	\item  A \Ls-definable  set   $A\subseteq \R$ is
 	periodic if and only if it is of the form $p\Z+B$ where $p \in \Q$ and $B\subseteq [0,p)$ is a 
 	finite union of intervals with rational endpoints. 
 	\item For every  \Ls-definable  set $A\subseteq \R$ 
 	 there exist two periodic sets $A_1,A_2 \subseteq \R$ and two reals $m_1,m_2$ 
 	such that  
 	$A\cap [m_1,+ \infty) = A_1\cap [m_1,+ \infty)$ and   
 	$A \cap (- \infty, m_2] = A_2\cap (- \infty, m_2]$\footnote{Theorem 6.1 of \cite{Weis99}
 misses the case where the subset $A$ 
 has different right and
 left ultimate periods such as $-2 \N \cup 3 \N$.}.
 \end{enumerate}
\end{prop}

\begin{proof} $(1)$ 
Assume that $A$ is periodic of period $p$. Let $B=A\cap [0,p)$. Then $A= p\Z+B$. Since $B$ is \Ss-definable it is a finite union of
intervals included in $[0,p)$ with rational endpoints. The converse is trivial.

$(2)$ Let $A\subseteq \R$ be \Ls-definable. We prove the existence of $A_1$ and $m_1$ (the proof for $A_2$ and $m_2$ is similar).
By Theorem \ref{th:separation-integer-fractional} we have 
\begin{equation}
\label{eq:integer-fractional-A-B-C}
A=\bigcup^{K}_{k=1} (B_{k} + C_{k} )
\end{equation}
where all $B_{k} \subseteq \Z$  are disjoint
$\langle \Z,+,<\rangle$-definable subsets and all  $C_{k}\subseteq [0,1)$  are distinct \Ss-definable subsets.

By \cite{Pre29} for every $k$ there exist two integers $n_k,p_k \geq 0$ such that
\begin{equation}
\label{eq:periodZ}
\forall x \geq n_k \  x \in B_k \leftrightarrow x+p_k \in B_k.
\end{equation}
Observe that if (\ref{eq:periodZ}) holds for $n_k,p_k$ then it still holds for the pair of integers $m_1,p \geq 0$ 
where  $p=\text{lcm}\{n_j \mid 1 \leq j \leq K\}$ and $m_1$ is a sufficiently large  a multiple of $p$. Therefore for every $k \in \{1,\dots,K\}$ we have 
\begin{equation}
\label{eq:periodZunif}
\forall x \geq m_1 \ x \in B_k \leftrightarrow x+p \in B_k.
\end{equation}
This implies that there exist $K$ disjoint sets  $S_{k}\subseteq \{0, \ldots, p-1\}$ such that for every $k \in \{1,\dots,K\}$ we have  

\[
B_{k}\cap [m_1, \infty)= (p\Z + S_{k}) \cap [m_1, \infty).
\]
The claim of the proposition follows by setting 
$A_1= p\Z + B $ where $ B=\displaystyle\bigcup^{K}_{k=1} S_{k} + C_{k}\subseteq [0,p)$.
\end{proof}

\begin{lem} 
	\label{le:transfer-period}
With the notations of Proposition  \ref{pr:special-ultimately-periodic-sets}$(2)$,
a real $q$ is a right (resp. left) ultimate period  of the \Ls-definable subset $A$ if and only if it is a period of  $A_1$ (resp. $A_2$).
\end{lem}

\begin{proof}
We only give the proof for $A_1$ (the proof for $A_2$ can be handled 
similarly). It follows from the equivalences

\[
\begin{array}{ll}
\forall y\geq m_{1} &  y\in A \leftrightarrow  y\in A_1 \leftrightarrow y+q \in A_{1} \leftrightarrow y+q \in A\\
\forall y\geq m_{1}  &y\in A _{1}\leftrightarrow  y\in A \leftrightarrow y+q \in A \leftrightarrow y+q \in A_{1}.
\end{array}
\]
                \par \vspace{-1.1\baselineskip}\qedhere
\end{proof}

\begin{lem}
\label{le:minimal-period}
With the notations of Proposition  \ref{pr:special-ultimately-periodic-sets} $(2)$, if $A=p\Z+B$ with $\emptyset \subsetneq  B\subsetneq [0,p)$, then  
the set of periods of $A$ 
is a discrete cyclic subgroup of $\R$ whose elements are rational.
It is generated by its element of minimal positive absolute value.
\end{lem}

\begin{proof}
The set $P$ of periods of $A$ is clearly a subgroup. 
Let us prove that $P$ is discrete, i.e. that there cannot be arbitrarily small periods. Indeed, 
set $A=p\Z + B$ where $p \geq 0$,  $B$ is a finite union of intervals in $[0,p)$
and $\emptyset \subsetneq  B\subsetneq [0,p)$. Let $q>0$ be a period of $A$. 
We consider three exclusive cases:

\begin{enumerate}

\item if  $B=\{0\}$ or $B=(0,p)$, then the condition $q=0+q\in A $
implies $q\geq p$.

\item if both $B$ and $[0,p) \setminus B$ consist of a unique interval, then we have 
either $B=[a,p)$ with $a>0$ or $B=[0,b)$ with $b<p$. In the first case we have $q\geq a$ and in the 
latter case $q\geq p-b$.

\item if $B$ or $[0,p) \setminus B$ consist of at least two intervals:  since  a set and its complement have the same periods, we can assume without loss of generality that $B$ contains 
two disjoint and consecutive intervals with respective extremities  $a_{1},b_{1} $ and $ a_{2},b_{2}$ (the proof for the other case is similar).
If $b_{1}= a_{2}$ then $(a_{1},b_{1}) $ is right open and  $ (a_{2},b_{2})$ is left open.
Then the fact that $q$ is a period and the equality $b_{1}= a_{2}\in A $  imply  $b_{1}-q, a_{2}+q\in A$ thus $q\geq \max\{b_{1}-a_{1}, b_{2}-a_{2}\}$.
Now if $b_{1}<a_{2}$ then for all $a_{1}<y<b_{1} $ we have $y+q\in A$ which implies $q\geq a_{2}-b_{1}$. 
\end{enumerate}

We proved that $P$ admits a minimal positive element, say $p_0$. The fact that $P$ is cyclic  and generated by $p_0$ is well-known. In order to prove that $P \subseteq \Q$, it suffices to prove that $p_0 \in \Q$. Now the \Ls-formula
\[
\phi(p)= p>0 \wedge \forall x\ (x\in A \leftrightarrow x+p\in A)
\]
defines the set $P^+$ of positive elements of $P$, thus $p_0$ is \Ls-definable as the minimal element of $P^+$, and the result follows from Corollary \ref{cor:basicSsFacts}(2). 
\end{proof}

As a consequence of Lemmas \ref{le:transfer-period} and \ref{le:minimal-period}, and Proposition \ref{pr:special-ultimately-periodic-sets} we obtain the following result.

\begin{lem}
\label{le:minimal-right-period}
If $A \subseteq \R$ is \Ls-definable but not \Ss-definable then $A$ has either  a minimal ultimate right or a minimal ultimate left period.
\end{lem}

\begin{proof}
	Either $A \cap (-\infty,0]$ or $A \cap [0,\infty)$ is not \Ss-definable. Assume that the latter case holds, and let us prove that $A$ admits a minimal ultimate right period.  By Proposition \ref{pr:special-ultimately-periodic-sets} there exist a periodic subset $A_1$ and a real $m$ such that 
	$A\cap [m, \infty[ = A_1\cap [m, \infty[$. With the notations of Proposition  \ref{pr:special-ultimately-periodic-sets} we have $A_1=p\Z+B$  where $ B\subseteq [0,p)$. Now  the set $A \cap [0,\infty)$ is \Ls-definable but not \Ss-definable, thus by Lemma \ref{le:restriction-to-bounded-domain} it cannot have an upper bound, and the same holds for $A_1$. It follows that $B \ne \emptyset$.  Using a similar argument with the complement of $A$ leads to $B \ne [0,p)$.
	Thus by Lemmas  \ref{le:transfer-period} and \ref{le:minimal-period}, $A_1$ admits a minimal period which is also a minimal ultimate right period of $A$.
	
	Similarly if $A \cap (-\infty,0]$ is not \Ss-definable then $A$ admits a minimal ultimate left period.
\end{proof}

\begin{lem}
\label{le:multiples-of-period}
If $A\subseteq \R$ is ultimately right  periodic and admits a  minimal ultimate right period $p>0$, then
\begin{enumerate}
	\item For all reals $x,y$, if
	\begin{equation}
	\label{eq:defining-Z}
	\forall u \geq 0 \ (x+u\in A \leftrightarrow y+u \in A)
	\end{equation}
	then $p$ divides $x-y$.
	\item Conversely for every $z$ multiple of $p$ there exist $x,y$  satisfying (\ref{eq:defining-Z}) and
	$x-y=z$. 
\end{enumerate}
\end{lem}

\begin{proof}
$(1)$ 
It suffices to prove that 
 $x-y=q$ is an ultimate right period. We have
\[
\begin{array}{ll}
\forall v \geq x   \quad  v=x+(v-x)\in A \leftrightarrow y + (v-x)\in A \leftrightarrow v+q\in A.\\
\end{array}
\]

$(2)$
Let $m$ be such that 
\[
\forall t\geq m\ (t\in A \leftrightarrow t+p\in A).
\]
This condition implies that for all positive multiples $q$ of $p$ we have
\[
\forall t \geq m\ (t\in A \leftrightarrow t+q\in A).
\]
Set $x=m, y=m+q$ and $u=t-m$. The above condition is equivalent to
\[
\forall u\geq 0\ (x+u\in A \leftrightarrow y+u\in A).\qedhere
\]
\end{proof}

\subsection{Proof of Theorem \ref{th:dim1}}

\medskip Now we proceed by  induction on the dimension $n$.
Assume $n=1$. If  $X$ is definable in \Ls\  but not in \Ss\ then by Lemma \ref{le:minimal-right-period}, the set $X$ has either a minimal ultimate left or a minimal ultimate right period. Assume w.l.o.g that the latter case holds, i.e that $X$ has an  ultimate right period $p>0$. By Lemma \ref{le:minimal-period}, $p$ is a rational
number, say $p=\frac{a}{b}$.

Applying Lemma \ref{le:multiples-of-period} the subset $\Z$ can be defined in $\langle  \R, +,<,1,X \rangle$ by the formula
\[
\phi_ {X}(x)= \exists y, z\ \forall u \geq 0 \ ((y+u\in X \leftrightarrow z+u\in X) \wedge ax=b(y-z)).
\]

 Now we pass to the general case $n \geq 2$.
Let  $X \subseteq \R^n$ be \Ls-definable but not \Ss-definable. By Theorem \ref{th:CNS}, either some rational section of $X$ is not \Ss-definable,  or  $X$ admits infinitely many singular points. In the first case, the result follows from the induction hypothesis and the fact that every rational section of  $X$ is definable in  \str{\R,+,<,1, X} thus also in \Ls.

 In the second case, by Proposition \ref{pr:finitely-many-singular-points} the set  
  $X$ contains  a countably infinite number of singular points.
The set $S$ of singular points is $\langle  \R, +,<,1,X \rangle$-definable by the formula 
(\ref{eq:I-singular}) thus it is 
also  \Ls-definable and the same holds for any projection of $S$ over a component. 
Therefore some of the $n$ projections over the $n$ components is a \Ls-definable  subset of $\R$
and contains a countably infinite number of singular points
thus is not \Ss-definable, and we may apply case $n=1$.
\qed

\section{Yet another characterization}
\label{sec:combinatorial-characterization}

We provide in this section an alternative characterization of  \Ss-definability for \Ls-definable relations.

A line in $\R^n$ is  {\em rational } if it is  
the intersection of hyperplanes defined by equations with rational coefficients. Every rational line is \Ss-definable. 

\begin{thm}
	\label{th:intersections-with-lines}
	A \Ls-definable relation  $X\subseteq \R^{n}$ is  \Ss-definable if and only if the intersection of $X$ with every rational line is \Ss-definable.
\end{thm}

\begin{proof} The condition is necessary because  every rational line is \Ss-definable and the  set of \Ss-definable relations 
	is closed under intersection.  
	
	\bigskip Now we prove the converse.
By Proposition \ref{pr:CNS1} 	it suffices to prove that all rational restrictions have finitely many singular points.
Since all rational restrictions  of $X$ are \Ls-definable the result will follow from the next Lemma.
	
\begin{lem}
\label{le:}
Let $X\subseteq \R^{n}$ be  \Ls-definable relation  and assume the intersection of $X$ with every rational line is \Ss-definable. 
Then $X$ has finitely many singular points.
\end{lem}	

\begin{proof} 
	We prove that if $X$ has infinitely many singular points then there exists some rational line $l$ such that $l \cap X$ is not a finite union of segments, which contradicts the fact that $l \cap X$ is \Ss-definable. 
		
	Let us first prove that the set of singular points is unbounded. Assume for a contradiction that there exists a rational number $M>0$ such that $X \cap [0,M)^n$ contains all $X-$singular points. The set $X'=X \cap [0,M]^n$ is \Ss-definable by Proposition \ref{le:restriction-to-bounded-domain}, and moreover by definition every $X-$singular element of $[0,M)^n$ is also $X'-$singular, thus there exist infinitely many $X'-$singular points, which contradicts Theorem \ref{th:CNS}.
		
	An hypercube of the form $[a_{1}, a_{1}+1]\times \cdots \times [a_{n}, a_{n}+1]$
	with $a_{1},\ldots,  a_{n}\in \Z$ is called  \emph{elementary}. All $X$-singular points 
	which belong to the interior of some  elementary hypercube $H$ are $(X\cap H)$-singular 
	and the converse is true. Observe that a point may be $X$-singular without being $(X\cap H)$-singular
	{ if it belongs to the boundary of the hypercube}.
	In order to avoid this problem 
	we consider an elementary cube surrounded by its $3^{n}-1$ neighbours. Then a point in a elementary 
	cube is $X$-singular if and only if it is singular in this enlarged hypercube.
	
	Formally, we start by extending the notion of integer part to vectors in $\R^{n}$ by setting 
	for every point  $x\in \R^{n}$ 
	\[
	\lfloor x \rfloor= (\lfloor x_{1} \rfloor ,  \ldots, \lfloor x_{n} \rfloor ).
	\]
	Set $\mathcal{S}=  \{-1,0,1\}^{n}$. For all $\sigma \in \mathcal{S}$ and $a\in \Z^{n}$ define
	\[
	\begin{array}{l}
	D_{\sigma} = \overbrace{[0,1]\times \cdots \times [0,1]}^{n \text{ times}} + \sigma, \ 
	H_{\sigma}(a)= (a + D_{\sigma})\cap X, \  L_{\sigma}(a) 
	= - a +  H_{\sigma}(a)
	\end{array}
	\]
	and 
	\[
	H(a)= \bigcup_{\sigma\in \mathcal{S}}  H_{\sigma}(a), \
	L(a)= \bigcup_{\sigma\in \mathcal{S}}  L_{\sigma}(a).
	\]

	Because of the decomposition of Theorem \ref{th:separation-integer-fractional},
	the set $\mathcal{L}$ of distinct subsets $L(a)$,
	called \emph{elementary neighborhood}, 
	is finite when $a$ ranges over $\Z^{n}$ since each $-\sigma + L_{\sigma}(a)$
	is equal to some $X^{(F)}_{k}$. 
	Furthermore, say that  $H(a)$ \emph{contains a singular point}  $x$ if $x$ is 
	$X$-singular and $x\in H_{(0,\ldots, 0)}(a)$.
	Since the set of singular points is unbounded, there exist a fixed $L\in \mathcal{L}$ and infinitely many $a$'s such that 
	$a+ L_{(0,\ldots, 0)}(a)$ contains a singular point. Furthermore, in each 
	$L_{(0,\ldots, 0)}(a)$ there   exist  at most finitely many 
	$X$-singular points.
	Consequently, there exists a fixed element $z\in \Q^{n}\cap [0,1]\times \cdots \times [0,1]$ and a fixed 
	elementary neighborhood $L$ such that for infinitely many $a\in \Z^{n}$ the point $a+z$ is $X$-singular and  
	$L_{a} = L$. 
	Consider the integer $K$ as defined in Theorem \ref{th:separation-integer-fractional}, and the mapping  $\iota:\mathcal{S}\mapsto \{1, \ldots, K\}$
	such that 
	\[
	L=  \bigcup_{\sigma\in \mathcal{S}}  \sigma + X^{(F)}_{\iota(\sigma)}. 
	\]
	The set of elements $a$ satisfying  
	$L_{a} = L$ is the intersection 
	\[
	\displaystyle \bigcap_{\sigma\in \mathcal{S}}  -\sigma + X^{(I)}_{\iota(\sigma)} \subseteq \Z^{n}
	\]
	which is  infinite and semilinear.
	Now, all infinite semilinear subsets contain a subset of the form
	$A=u+\N v$ where $u,v\in \Z^{n}$ with $v\not=0$.  
	
 	For some sufficiently small positive real $r$, for all $w\in A$  the sets 
	$-w + (B(w,r) \cap X)$ are identical.  Consider any element $w' \in A$, and assume further that $r \in \Q$. Since $w'$ is singular, $B(w',r)\cap X$ is not a union of lines parallel
   	to any direction and in particular to the direction 
   	$v$. Thus the set $Y_r$ of points $t \in \R^n$ such that $|t|<r$ and $B(w',r)\cap L_v(w'+t)$ intersects both $X$ and its complement is not empty. Now both $w'$ and $r$ have rational components, and $X$ is \Ls-definable, thus $Y_r$ is also \Ls-definable, and since it is bounded it is also \Ss-definable by Lemma \ref{le:restriction-to-bounded-domain}. By Corollary \ref{cor:basicSsFacts}(2), $Y_r$ contains an element $t\in \Q^n$. Now for all $w\in A$  the sets
	$-w + (B(w,r) \cap X)$ are identical, and $|t|<r$, thus for every $w \in A$ the segment $B(w,r)\cap L_v(w+t)$ intersects both $X$ and its complement. It follows that the line $\ell=L_v(w'+t)$ (which coincides with all lines $L_v(w+t)$ for $w \in A$) is such that $l \cap X$ is not a finite union of segments, and thus by Corollary \ref{cor:basicSsFacts}(1) cannot be \Ss-definable. 	   	
\end{proof}
This concludes the proof of Theorem \ref{th:intersections-with-lines}.
\end{proof}

\section{Conclusion}
\label{sec:conclusion}

We discuss some extensions and open problems.

It is not difficult to check that the main arguments used to prove Theorems \ref{th:eff}, \ref{th:dim1} and \ref{th:intersections-with-lines}, still hold if one replaces $\R$ with $\Q$. Observe that by \cite{Mil01}, \str{\Q, +, <, \Z} and \Ls\ are elementary equivalent structures.

Is it possible to remove our assumption that $X$ is \Ls-definable in Theorem \ref{th:CNS}? We believe that the  answer is positive\footnote{{we proved this claim after the present paper was written (preprint available at \url{https://arxiv.org/abs/2102.06160}).}}.  
Note that even if one proves such a result, the question of providing an effective characterization is more complex. Indeed the sentence $\psi_n$ of Proposition \ref{prop:autodef} expresses a variant of the criterion of Theorem \ref{th:CNS}, and we use heavily the fact that we work within \Ls \, to ensure that this variant is actually equivalent to the criterion.

In particular the construction of $\psi_n$
relies on Lemma \ref{le:equivalence-to-singularity} to express that a point is $X-$singular. 
However if we consider, e.g., $X=\Q$ then every element $x$ of $X$ is singular while no element $x$ of $X$ satisfies the condition stated in Lemma \ref{le:equivalence-to-singularity}. 

Another question is the following. In Presburger arithmetic it is decidable whether or not a formula is equivalent to a  formula in the structure without $<$, cf. \cite{ChoFri}. 
What about the case of the structure \Ls\ ?


\section*{Acknowledgment}

We wish to thank the anonymous referee and Erich Gr\"adel for useful suggestions.

\bibliographystyle{alpha}

\bibliography{rz-lmcs}

\end{document}